%% file: fairpowercontrol_double_final.tex
\documentclass[journal,comsoc]{IEEEtran}
\usepackage[algo2e]{algorithm2e}
\usepackage[super]{nth}
\usepackage{cite}
\usepackage{graphicx}  
\usepackage{amsmath}  
\usepackage{algorithm}
\usepackage[latin1]{inputenc}
\usepackage[T1]{fontenc}
\usepackage{times}
\usepackage{tabularx}
\usepackage{stfloats}
\usepackage{fancybox}
\usepackage{verbatim}
\usepackage{longtable}
\usepackage{url}
\usepackage{boxedminipage}
\usepackage{amssymb}
\usepackage{tabu}
\usepackage{pdfpages}
\usepackage{epstopdf}
\usepackage{cite}
\usepackage{algpseudocode}
\usepackage{optidef}
\usepackage{lscape}
\usepackage{xcolor}  
\usepackage{subcaption}
\usepackage[font=small,labelfont=bf]{caption}
\usepackage{mathtools}
\usepackage{float}

\input{commands}

\IEEEoverridecommandlockouts

\begin{document}

\title{Enhanced Fairness and Scalability of Power Control Schemes in Multi-Cell Massive MIMO}

\author{Amin Ghazanfari,~\IEEEmembership{Student Member,~IEEE,} Hei Victor Cheng, Emil Bj{\"o}rnson,~\IEEEmembership{Senior Member,~IEEE} and Erik G. Larsson,~\IEEEmembership{Fellow,~IEEE}\\
	\thanks{A preliminary version of this work was presented at the IEEE International Conference on Acoustics, Speech and Signal Processing (ICASSP) 2019 \cite{amin2019icassp}. This paper was supported by ELLIIT and by the European Union's Horizon 2020 research and innovation programme under grant agreement No 641985 (5Gwireless).}
	\thanks{Amin Ghazanfari, Emil Bj\"{o}rnson and Erik G. Larsson are with the Department of Electrical Engineering (ISY), Link\"{o}ping University, 581~83 Link\"{o}ping, Sweden (email:\{amin.ghazanfari, emil.bjornson, erik.g.larsson\}@liu.se).}
	\thanks{H. V. Cheng is with The Edward S. Rogers Sr. Department of Electrical
		and Computer Engineering, University of Toronto, Toronto ON M5S 3G4,
		Canada (email:hvc@ieee.org).}
}

\maketitle
\begin{abstract}
	This paper studies the transmit power optimization in multi-cell massive multiple-input multiple-output (MIMO) systems. Network-wide max-min fairness (NW-MMF) and network-wide proportional fairness (NW-PF) are two well-known power control schemes in the literature. The NW-MMF focus on maximizing the fairness among users at the cost of penalizing users with good channel conditions. On the other hand, the NW-PF focuses on maximizing the sum SE, thereby ignoring fairness, but gives some extra attention to the weakest users. However, both of these schemes suffer from a scalability issue which means that for large networks, it is highly probable that one user has a very poor channel condition, pushing the spectral efficiency (SE) of all users towards zero.	
	To overcome the scalability issue of NW-MMF and NW-PF, we propose a novel power control scheme that is provably scalable. This scheme maximizes the geometric mean (GM) of the per-cell max-min SE. To solve this new optimization problem, we prove that it can be rewritten in a convex optimization form and then solved using standard tools. The simulation results highlight the benefits of our model which is balancing between NW-PF and NW-MMF.  
\end{abstract}

\begin{keywords}
	Power control, Massive MIMO, Fairness, Optimization.
\end{keywords}

\section{Introduction}
\label{sec:intro}
Massive multiple-input multiple-output (MIMO) \cite{marzetta2010noncooperative} is a key technology in 5G \cite{Parkvall2017a,WWB5G}. It refers to a system in which the cellular base stations (BSs) are equipped with very many antennas. Massive MIMO supports spatial multiplexing of many users, coherent beamforming, and spatial interference mitigation. It enhances the spectral and energy efficiency compared with conventional MIMO setups. Unlike conventional cellular systems with one or a few antennas per BS, the power control in massive MIMO systems benefits from channel hardening, namely that the small-scale fading average out when having many antennas per BS \cite{bjornson2017massive}. It means that in massive MIMO, one can optimize the transmission power based on only the large-scale fading coefficients and spatial correlation, instead of optimizing with respect to the small-scale fading coefficients, which changes rapidly and would require very rapid power control updates. Power control schemes with different utility functions have been considered in the massive MIMO literature \cite{bjornson2017massive,guo2014uplink,kammoun2014low,zhao2013energy,wu2016asymptotically,guo2016security,zhang2015power,liu2017pilot1,assaad2018power,baracca2018downlink,hao2017power}. In particular, max-min fairness (MMF) is a classical utility function that has been studied for different setups in \cite{van2016joint,cheng2017optimal,van2017joint,yang2017massive,xiang2014massive,zarei2017max,akbar2018downlink}. It provides the same quality of service at all user locations, which is a highly desirable feature in future systems.

Unfortunately, applying a network-wide MMF (NW-MMF) utility to a multi-cell massive MIMO network leads to a scalability issue since the performance is limited by the weakest user. When increasing the number of cells and active users in the network, the probability of having a user with an extremely poor channel gets higher due to shadow fading. Therefore, NW-MMF optimization leads to the situation in which all users in the network suffer from the weak channel of the worst user. In other words, the larger the network, the lower the per-user spectral efficiency (SE) becomes when optimizing for NW-MMF. Increasing the number of cells to infinity will eventually result in zero SE for all users in the network. This is a major problem that was pointed out in the textbooks \cite{redbook,bjornson2017massive}, but seldom discussed in scientific papers where the simulation setups are often too small to observe overly small SEs when considering the NW-MMF utility. In conclusion, the NW-MMF schemes proposed in the literature are unsuitable for providing fairness in practical cellular networks. The network-wide proportional fairness (NW-PF) is another well-known utility function for power control in multi-cell massive MIMO \cite{bjornson2017massive}. It has the benefit of balancing between sum SE optimization and MMF. In NW-PF, the optimization objective is defined as the product of SINRs of all users in the network. Therefore, this scheme suffers from the same scalability issue as NW-MMF.

\subsection{Related works and contributions}
In this paper, we propose a rigorous optimization framework that provides network-wide power control for multi-cell Massive MIMO in which per-cell max-min fairness is guaranteed. A heuristic approach to resolve the scalability issue of NW-MMF was considered in \cite[Ch.~6]{redbook}. The idea is to maximize the minimum SE within each cell and let the cells have distinct SEs. This is done by neglecting the coherent interference (i.e. the interference from contaminating cells that are sharing the same pilot sequence with the desired cell) and allowing all the cells to utilize their full powers and applying MMF within the cells and then compensate to the first order of approximation for the effects of coherent interference. Hence, the weak users have a lower impact on the whole network performance and mostly affect their own cells. 
The proposed algorithm in \cite{redbook} is computationally efficient, but relies on approximations and there is no guarantee of optimality.
Inspired by this algorithm, we are proposing a new utility function that can be optimized rigorously: maximization of the geometric mean (GM) of the max-min SEs in each of the cells. We also generalize the heuristic algorithm from \cite{redbook} to handle correlated fading channels. The NW-PF utility was considered in \cite{bjornson2017massive} to balance between sum SE optimization and fairness. In simulations, it outperforms NW-MMF in terms of SE for most users, but it gives no fairness guarantees except for giving non-zero SE to every user if we have non-zero channels for all users. Recall that, NW-PF also gives nothing but zero SE for all users in case one of the users suffers from deep fading with a zero-valued channel.

The main contributions of this paper are:
\begin{itemize}	
	\item We formulate the novel GM per-cell MMF power control problem, which is proved to be scalable in terms of performance. We then reformulate the problem to reach a convex formulation that can be solved to global optimality in an efficient way. 
	The new scheme outperforms the heuristic scheme in \cite{redbook} in terms of per-user SE and sum SE of the network in some cases and gives a comparable performance in other cases.
	\item We provide a rigorous mathematical proof of that the NW-MMF and NW-PF power control schemes are not scalable for multi-cell Massive MIMO systems.
	\item To further investigate the benefits of the proposed power control scheme, we define and solve two more power control schemes for the problem at hand: NW-MMF and NW-PF. The numerical results show that the proposed power control scheme combines the benefits of NW-MMF and NW-PF without suffering from the scalability issue of NW-MMF and NW-PF.
	\item We consider correlated Rayleigh fading and also explain how the framework can be used with other channel models. We extended the heuristic scheme proposed in \cite{redbook} for correlated Rayleigh fading channels. This scheme provides an approximate solution of the GM per-cell MMF power control scheme. 
	\item We investigate the effect of different pilot reuse factors on the performance of our proposed power control scheme.
\end{itemize}

To illustrate the scalability issue of NW-MMF and NW-PF, we provide the following example. We consider a simulation setup consisting of $16$ cells and $2$ users per cell that are communicating in the uplink (UL) with correlated Rayleigh fading channels (for other simulation parameters please refer to Section \ref{result}). For one of the users in the network (user A), the large-scale fading coefficient is manually fixed to some arbitrary values indicated in ${\rm dB}\,$ scale in Fig. \ref{fig:BetaVary}. The figure shows the sum SE achieved with NW-MMF, NW-PF, and our proposed GM per-cell MMF power control problems.  We can see that when the large-scale fading coefficient value for user A is very low, the sum SE for both NW-MMF and NW-PF is zero (given the numerical precision in our computation) while our approach still achieves a high sum SE.
	
\begin{figure}[htb!]
	\centering
	\includegraphics[width=.8\columnwidth]{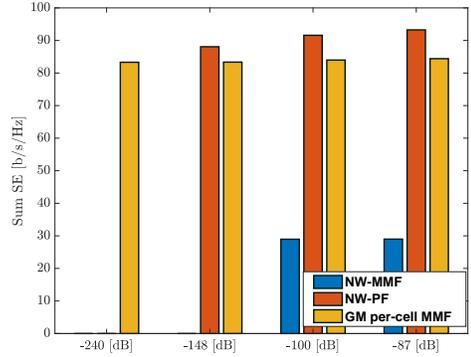}
	\caption{Sum SE for UL data transmission with different large-scale fading values for user A.}
	\label{fig:BetaVary}
\end{figure}

Note that in \cite{amin2019icassp}, which is the conference version of the current paper, we proposed the novel GM per-cell MMF power control for the case of uncorrelated Rayleigh fading channel and assumed that the pilots are reused in every cell. In this paper, we consider correlated Rayleigh fading and arbitrary pilot reuse sets, and we also demonstrate that the optimization framework can be used in many other scenarios.

\textbf{Notation:} We use boldface lower case to indicate column vectors $\mathbf{x}$, and boldface upper case is used for matrices, $\mathbf{X}$. An identity matrix with size $M$ is denoted as $\mathbf{I}_M$. The conjugate transpose of $\mathbf{X}$ is denoted as $\mathbf{X}^{{\rm H}}$.  In addition, the operator $\mathbb{E}\{\boldsymbol{\cdot}\}$ denotes the expectation of a random variable. The notation $\|\mathbf{x}\|$ stands for the L$_2$-norm of the vector $\mathbf{x}$ and $\mathrm{tr}\left(\mathbf{A}\right)$ is the trace
of a square matrix $\mathbf{A}$. The Hadamard product of $\mathbf{A}$ and $\mathbf{B}$ is denoted by $\mathbf{A}\odot \mathbf{B}$. The notation $ \CN({\mathbf{0}},{\mathbf{R}})$ is used to show the circularly symmetric complex Gaussian distribution with zero mean and correlation matrix ${\mathbf{R}}$.

\section{System Model} 
\label{SystemModel}
In this paper, we consider a multi-cell massive MIMO setup that consists of $L$ cells, each associated with one BS. Each BS is equipped with $M$ antennas and is serving $K$ single-antenna users. In the considered setup, the channel response between BS $l$ and user $k$ in cell $l'$ is defined as ${\mathbf{h}}^{l}_{l'k}\sim \CN({\mathbf{0}},{\mathbf{R}}^{l}_{l'k})$, where ${\mathbf{R}}^{l}_{l'k} \in \mathbb{C}^{M\times M}$ is the positive semi-definite spatial correlation matrix of the channel. We define $\beta^{l}_{l'k}  = \frac{\mathrm{tr}\left(\mathbf{R}^{l}_{l'k}\right)}{M}$, where $\beta^{l}_{l'k} \geq 0$ is the corresponding average large-scale fading coefficient among the antennas. We use conventional block fading to model the randomness of the wireless channels over time and frequency \cite{redbook}. The coherence block of a channel is defined as the time-frequency block in which the channel is constant. The channels change independently from one block to another according to a stationary ergodic random process. The number of samples per coherence block is given by $\tau_c = T_{c} B_{c}$, where $T_{c}$ is the coherence time and $B_{c}$ is the coherence bandwidth \cite[Ch.~2]{redbook},\cite[Ch.~2]{bjornson2017massive}. Therefore, it is assumed that channel estimation is carried out at each BS once per coherence block. Each user transmits a pilot sequence from a predefined set of mutually orthogonal pilots.  It is assumed that $\tau_p$ samples (with $\tau_p\leq \tau_c$) are dedicated for pilot transmission and the remaining samples will be utilized for UL and downlink (DL) data transmission. Since the complexity of MMSE estimation is prohibitive \cite[Ch.~3]{bjornson2017massive}, we assume that the BSs apply the element-wise MMSE (EW-MMSE) channel estimator and the channel estimate of the channel response between BS $l$ and user $k$ in cell $l'$ is \cite{ozdogan2018massive}
\begin{equation}
\begin{aligned}
\hat{\mathbf{h}}^{l}_{l'k} \sim \CN({\mathbf{0}},{\mathbf{\Sigma}}^{l}_{l'k}),
\end{aligned}
\end{equation}
where ${\mathbf{\Sigma}}^{l}_{l'k} = p_{\rm ul} \tau_p {\mathbf{D}}^{l}_{l'k} {\mathbf{\Lambda}}^{l}_{l'k} \left({\mathbf{\Psi}}^{l}_{l'k}\right)^{-1}{\mathbf{\Lambda}}^{l}_{l'k} {\mathbf{D}}^{l}_{l'k}$. In addition, ${\mathbf{\Psi}}^{l}_{l'k}$, ${\mathbf{D}}^{l}_{l'k}$ and ${\mathbf{\Lambda}}^{l}_{l'k}$ are defined as
\begin{equation}
\begin{aligned}
{\mathbf{\Psi}}^{l}_{l'k} = \left(\sum_{j\in \mathcal{P}_{l'}} p_{\rm ul} \tau_{p} {\mathbf{R}}^{l}_{jk} + {\mathbf{I}}_{M}\right)^{-1},
\end{aligned}
\end{equation}

\begin{equation}
\begin{aligned}
\mathbf{D}^{l}_{l'k} = \mathbf{R}^{l}_{l'k} \odot \mathbf{I}_{M},
\end{aligned}
\end{equation}
\begin{equation}
\begin{aligned}
\mathbf{\Lambda}^{l}_{l'k} =\left(\left[\sum_{j\in \mathcal{P}_{l'}} p_{\rm ul}\tau_p \mathbf{R}^{l}_{jk} + \mathbf{I}_{M} \right] \odot \mathbf{I}_{M}\right)^{-1},
\end{aligned}
\end{equation}
where the set $\mathcal{P}_{l'}  \subset \{j: j = 1,\dots,L \}$ of all cells in the network where the users have the same set of pilots as in cell $l'$. If two cells use the same set of pilot sequences, then user $k$ in these cells use same pilot sequence.
In addition, $\rho_{\rm ul}$ is the normalized UL transmit power and we have independent additive noise at the BSs with i.i.d.~elements distributed by $\CN({{0}},1).$
Note that each BS only needs to know the diagonals of the spatial correlation matrices to apply the EW-MMSE estimator. The diagonals are relatively easy to estimate while the full correlation matrices that are needed for the MMSE estimator are complicated to acquire \cite[Ch.~3]{bjornson2017massive}. 
We further assume that each BS performs maximum ratio (MR) processing during the data transmission phase. This assumption is made to obtain closed-form SE expressions for both UL and DL data transmission \cite{bjornson2017massive}. However, our method can be generalized to arbitrary linear processing techniques and channel models, which will be presented in Section \ref{otherChannelModels}. For the UL data transmission, the ergodic SE of user $k$ in cell $l$ is given by \cite[Th.~4.4]{bjornson2017massive}
\begin{equation} \label{eq:SE_ul}
\mathrm{SE}^{\rm ul}_{lk} = \left(1- \frac{\tau_p+\tau_d}{\tau_c}\right)\log_2\left(1+ \mathrm{SINR}^{\rm ul}_{lk}\left(\left\{\eta_{lk}\right\}\right)\right),
\end{equation}
where $\tau_d$ is the number of samples used for DL data transmission.
In addition, $\eta_{lk} \in [0,1]$ is the power control coefficient of user $k$ in cell $l$ and these will be optimization variables in this paper. We let $\{\eta_{lk}\}$ denote the set of all power control coefficients. The UL effective SINR  for user $k$ in cell $l$ is denoted as  $\mathrm{SINR}^{\rm ul}_{lk}\left(\left\{\eta_{lk}\right\}\right)$ and provided in \eqref{eq:ul_sinr_corr} at the top of the next page \cite{ozdogan2018massive}.
\begin{figure*}[t]
\begin{equation}\label{eq:ul_sinr_corr}
\begin{aligned}
\mathrm{SINR}^{\rm ul}_{lk} \left(\left\{\eta_{lk}\right\}\right)=\frac{\eta_{lk}{\rho_{\rm ul}} \left(\rho_{\rm ul}\tau_{p} 
	\mathrm{tr}\left(\mathbf{D}^{l}_{lk}\mathbf{\Lambda}^{l}_{lk}\mathbf{D}^{l}_{lk}\right)\right)^2}{\sum\limits_{l' =1
	}^{L} \sum\limits_{k'=1}^{K} \eta_{l'k'}\rho_{\rm ul} \chi^{lk}_{l'k'} - \eta_{lk}\rho_{\rm ul}\left(\rho_{\rm ul}\tau_{p} \mathrm{tr}\left(\mathbf{D}^{l}_{lk} \mathbf{\Lambda}^{l}_{lk}\mathbf{D}^{l}_{lk}\right)\right)^2 + \mathrm{tr}\left(\mathbf{\Sigma}^{l}_{lk}\right)},
\end{aligned}
\end{equation}  
where $\chi^{lk}_{l'k'}$ is defined as
\begin{equation}
\chi^{lk}_{l'k'} =  \mathrm{tr}\left(\mathbf{R}^{l}_{l'k'}\mathbf{\Sigma}^{l}_{lk}\right)+ \begin{cases}
\rho_{\rm ul}^2 \tau^{2}_{p} \left(\mathrm{tr}\left(\mathbf{D}^{l}_{l'k'}\mathbf{\Lambda}^{l}_{lk}\mathbf{D}^{l}_{lk}\right)\right)^2, & l' \in \mathcal{P}_{l} \textrm{ and } k' = k,\\
0, & \textrm{otherwise}.
\end{cases}
\end{equation}
\hrulefill
\end{figure*}

Similar assumptions are made for the case of DL data transmission, in which we have MR processing at the BSs with correlated Rayleigh fading channel and assuming that each BS uses the EW-MMSE estimator. The ergodic SE of user $k$ in cell $l$ is
\begin{equation}\label{eq:SE_dl}
\mathrm{SE}^{\rm dl}_{lk} = \left(1- \frac{\tau_p+\tau_u}{\tau_c}\right)\log_2\left(1+ \mathrm{SINR}^{\rm dl}_{lk}\left(\left\{\eta_{lk}\right\}\right)\right),
\end{equation}
where $\tau_u$ is the number of samples used for UL data transmission and $\mathrm{SINR}^{\rm dl}_{lk}$ is the effective SINR of user $k$ in cell $l$ and given in \eqref{eq:dl_sinr_corr} at the top of the next page \cite{ozdogan2018massive},
\begin{figure*}[t]

\begin{equation}\label{eq:dl_sinr_corr}
\begin{aligned}
\mathrm{SINR}^{\rm dl}_{lk}\left(\left\{\eta_{lk}\right\}\right)=\frac{\eta_{lk}{\rho_{\rm dl}} \left(\rho_{\rm ul}\tau_{p} 
	\mathrm{tr}\left(\mathbf{D}^{l}_{lk}\mathbf{\Lambda}^{l}_{lk}\mathbf{D}^{l}_{lk}\right)\right)^2}{\mathrm{tr}\left(\mathbf{\Sigma}^{l}_{lk}\right)\left(\sum\limits_{l' =1
	}^{L} \sum\limits_{k'=1}^{K} \eta_{l'k'}\rho_{\rm dl} \frac{\zeta^{lk}_{l'k'}}{\mathrm{tr}\left(\mathbf{\Sigma}^{l'}_{l'k'}\right)} - \eta_{lk}\rho_{\rm dl}\frac{\left(\rho_{\rm ul}\tau_{p} \mathrm{tr}\left(\mathbf{D}^{l}_{lk} \mathbf{\Lambda}^{l}_{lk}\mathbf{D}^{l}_{lk}\right)\right)^2}{\mathrm{tr}\left(\mathbf{\Sigma}^{l}_{lk}\right)} + 1\right)},
\end{aligned}
\end{equation} 
where $\zeta^{lk}_{l'k'}$ is defined as
\begin{equation}
\zeta^{lk}_{l'k'} =  \mathrm{tr}\left(\mathbf{R}^{l'}_{lk}\mathbf{\Sigma}^{l'}_{l'k'}\right)+ \begin{cases}
\rho_{\rm ul}^2 \tau^{2}_{p} \left(\mathrm{tr}\left(\mathbf{D}^{l'}_{lk}\mathbf{\Lambda}^{l'}_{l'k'}\mathbf{D}^{l'}_{l'k'}\right)\right)^2, & l' \in \mathcal{P}_{l} \textrm{ and } k' = k,\\
0, & \textrm{otherwise},
\end{cases}
\end{equation}
\hrulefill
\end{figure*}
where $\rho_{\rm dl}$ in \eqref{eq:dl_sinr_corr} is DL normalized transmit powers and $\eta_{lk} \in [0,1]$ is the power control coefficient of user $k$ in cell $l$. Note that the power control coefficients generally take different values in the UL and DL.

Uncorrelated Rayleigh fading is a special case of this model with ${\mathbf{R}}^{l}_{l'k} = \beta^{l}_{l'k} {\mathbf I}_{M}$. Therefore, the channel response between BS $l$ and user $k$ in cell $l'$ becomes ${\mathbf h}^{l}_{l'k}\sim \CN({\mathbf{0}},\beta^{l}_{l'k} {\mathbf I}_{M})$. 
Since the EW-MMSE and MMSE estimators coincide in this case, the channel estimation phase follows the standard minimum mean square error (MMSE) estimation approach in the literature, e.g., \cite{redbook,bjornson2017massive,kay1993fundamentals} and the derivation is omitted here. Hence, the MMSE estimate of ${\mathbf h}^{l}_{l'k}$ is denoted as ${\hat{\mathbf h}}^{l}_{l'k}\sim \CN({\mathbf{0}},\gamma^{l}_{l'k} {\mathbf I}_{M})$, where $\gamma^{l}_{l'k}$ is the corresponding variance:
\begin{equation}
\begin{aligned}
\gamma^{l}_{l'k}=  \frac{{\tau_p p_{\rm ul}\left({\beta^{l}_{l'k}}\right)^2}}{{1+ \tau_p \rho_{\rm ul}\sum\limits_{l''\in \mathcal{P}_{l'}}\beta^{l}_{l''k}}}.
\end{aligned}
\end{equation}
Note that if two BSs are sharing pilots, user $k$ in the respective cells use identical pilots for $k=1,\ldots,K$.

By using the mentioned information for uncorrelated Rayleigh fading and the assumption of MR processing at each BS for the UL data transmission, the ergodic SE of user $k$ in cell $l$ is identical to \eqref{eq:SE_ul} with a different SINR expression. 
In this case, $\mathrm{SINR}^{\rm ul}_{lk}$ is given as
\begin{equation}\label{eq:ul_sinr}
\begin{aligned}
&\mathrm{SINR}^{\rm ul}_{lk} \left(\left\{\eta_{lk}\right\}\right)=\\
&\frac{M \rho_{\rm ul} \gamma^{l}_{lk}\eta_{lk}}{1+ \sum\limits_{l'=1}^{L} \sum\limits_{k'=1}^{K}\rho_{\rm ul} \beta^{l}_{l'k'}\eta_{l'k'}+\sum\limits_{l'\in \mathcal{P}_{l}\backslash \{l\}}{}M\rho_{\rm ul}\gamma^{l}_{l'k}\eta_{l'k}}.
\end{aligned}
\end{equation}

For the DL data transmission with uncorrelated Rayleigh fading channels, the effective DL SINR for the case of MR processing at the BSs is 
\cite[Ch.~4]{redbook}
\begin{equation}\label{eq:dl_sinr}
\begin{aligned}
&\mathrm{SINR}^{\rm dl}_{lk} \left(\left\{\eta_{lk}\right\}\right)=\\
&\frac{M \rho_{\rm dl} \gamma^{l}_{lk}\eta_{lk}}{1+ \sum\limits_{l'=1}^{L}\rho_{\rm dl}\beta^{l'}_{lk} \left(\sum\limits_{k'=1}^{K} \eta_{l'k'}\right)+\sum\limits_{l'\in \mathcal{P}_{l} \backslash \{l \}}^{}M\rho_{\rm dl}\gamma^{l'}_{lk}\eta_{l'k}},
\end{aligned}
\end{equation}
which can be plugged into \eqref{eq:SE_dl} that provides the ergodic SE of user $k$ in cell $l$.

Note that, the SINR expressions for both cases of uncorrelated and correlated fading channel for both UL and DL data transmission with MR processing at the BSs can be written in a general form
\begin{equation}\label{eq:SINR_general}
\begin{aligned}
\mathrm{SINR}_{lk} \left(\left\{\eta_{lk}\right\}\right)= \frac{a_{lk}\eta_{lk}}{\sum\limits_{l'=1}^{L}\sum\limits_{k' = 1}^{K} b^{l'k'}_{lk} \eta_{l'k'} + \sum\limits_{l'\in \mathcal{P}_{l}\backslash \{ l \}} {c}^{l'}_{lk} \eta_{l'k} + d_{lk}},
\end{aligned}
\end{equation}
where $a_{lk}$, $b^{l'k'}_{lk}$, $c^{l'}_{lk}$ and $d_{lk}$ can be defined based on the corresponding channel model assumption and data transmission direction. Table \ref{coeff} summarizes the corresponding definitions of the aforementioned parameters for both UL and DL data transmission for correlated and uncorrelated Rayleigh fading channels.

\begin{table}
\begin{center}
	\caption{Explicit definition of parameters.} 
	\begin{tabular}{ | m{1cm} | m{2cm}| m{4.2cm} | } 
		\hline
		 & Uncorrelated & Correlated \\ 
		\hline
		\hline
		Uplink &  $a_{lk} = M \rho_{\rm ul} \gamma^{l}_{lk}$ \newline $b^{l'k'}_{lk} = \rho_{\rm ul} \beta^{l}_{l'k'}$ \newline $c^{l'}_{lk} = M \rho_{\rm ul} \gamma^{l}_{l'k}$ \newline $d_{lk} = 1$& $a_{lk}={\rho^3_{\rm ul}} \tau_{p}^2 
		{\mathrm{tr}}\left(\mathbf{D}^{l}_{lk}\mathbf{\Lambda}^{l}_{lk}\mathbf{D}^{l}_{lk}\right)^2$ \newline $b^{l'k'}_{lk} = \rho_{\rm ul} \chi^{\rm ul}_{l',k'}$\newline $c^{l'}_{lk} =\rho_{\rm ul}^2 \tau^{2}_{p} \mathrm{tr}\left(\mathbf{D}^{l}_{l'k}\mathbf{\Lambda}^{l}_{lk}\mathbf{D}^{l}_{lk}\right)^2$ \newline $d_{lk} = {\mathrm{tr}}\left(\mathbf{\Sigma}^{l}_{lk}\right)$ \\ 
		\hline
		Downlink &   $a_{lk} = M \rho_{\rm dl} \gamma^{l}_{lk}$ \newline $b^{l'k'}_{lk} = \rho_{\rm dl} \beta^{l'}_{lk}$ \newline $c^{l'}_{lk} = M \rho_{\rm dl} \gamma^{l'}_{lk}$ \newline $d_{lk} = 1$& $a_{lk}=\rho_{\rm dl} \rho^2_{\rm ul}\tau_{p}^2 
		{\mathrm{tr}}\left(\mathbf{D}^{l}_{lk}\mathbf{\Lambda}^{l}_{lk}\mathbf{D}^{l}_{lk}\right)^2$ \newline $b^{l'k'}_{lk} = \rho_{\rm dl} \frac{\chi^{\rm ul}_{l'k'}{\mathrm{tr}}\left(\mathbf{\Sigma}^{l}_{lk}\right)}{{\mathrm{tr}}\left(\mathbf{\Sigma}^{l'}_{l'k'}\right)}$\newline $c^{l'}_{lk} =\rho_{\rm dl}^2 \frac{{\mathrm{tr}}\left(\mathbf{R}^{l'}_{lk}\mathbf{\Sigma}^{l}_{l'k'}\right)}{{\mathrm{tr}}\left(\mathbf{\Sigma}^{l'}_{l'k'}\right)}$ \newline $d_{lk} = {\mathrm{tr}}\left(\mathbf{\Sigma}^{l}_{lk}\right)$ \\ 
		\hline
	\end{tabular}\label{coeff}
\end{center}
\end{table}

\section{Problem Formulation} 
\label{probelms}
This section motivates and defines the problem formulation. Specifically, we formulate and solve a new multi-cell MMF power control problem.
In order to evaluate and compare our proposed scheme with the state-of-the-art, we also define and solve two additional optimization problems. In the following subsection, we introduce our proposed method that is solving the GM of per cell MMF problem. This method resolves the scalability issue that happens when applying NW-MMF and NW-PF. These methods and the scalability issue are explained in detail in Subsections \ref{NW-MMF formulation} and \ref{NW-PF}.
\subsection{Proposed: Geometric-mean per-cell max-min fairness}
\label{GM PC-MMF formulation}
To solve the scalability issue of NW-MMF while keeping the focus on user fairness, we formulate a new optimization problem in which the optimization objective is the GM of per-cell max-min SE \footnote{In this work, we provide fairness by using the product of SEs in the objective function because the SEs are the operationally meaningful physical quantities. Moreover, using SEs will benefits the users with weaker channels as $y=\log_2(1+x)$ decreases slower than $y=x$. Therefore the users with low SINRs are penalized less in the optimization procedure.}. This optimization problem achieves MMF locally in each cell and performs proportional fairness between the cells in the network. Therefore, solving this optimization problem offers the benefit that a user in deep fading condition mainly affects the SE of its serving cell, but not the whole network. Hence, this problem does not suffer from the scalability issue. The optimization problem for the UL data transmission is
\begin{equation}\label{eq:opt_problem1}
\begin{aligned}
& \underset{\{t_l\},\{\eta_{lk}\}}{\mathrm{maximize}}
& &  \prod_{l=1}^L \log_2\left(1+\epsilon+t_l\right)   \\
& {\mathrm{subject~to}}
& &  0 \leq \eta_{lk} \leq 1, \forall~l,k, \\
& & & {\mathrm{SINR}}^{\rm ul}_{lk} \left(\left\{\eta_{lk}\right\}\right) \geq t_l, \forall~l,k,\\
\end{aligned}
\end{equation}
where $t_l$ is the minimum SINR of cell $l$ and $\epsilon> 0$ is a small control parameter. Note that ${\mathrm{SINR}}^{\rm ul}_{lk} \left(\left\{\eta_{lk}\right\}\right)$ can be expressed by either \eqref{eq:ul_sinr_corr}  or \eqref{eq:ul_sinr} depending on the channel model assumption. 
\begin{lemma}
	The control parameter $\epsilon>0$ prevents the utility function of \eqref{eq:opt_problem1} from being identically zero when at least one cell $l$ has a user with $\underset{k}{\mathrm{min}}({\beta_{lk}^{l}}) = 0$. Hence, the cells $l'$ with $\underset{k}{\mathrm{min}}({\beta_{l'k}^{l'}}) > 0$ will have $t_{l'}>0$ at the optimal point to \eqref{eq:opt_problem1}. 
\end{lemma}
\begin{proof}
	Cells with $\underset{k}{\mathrm{min}}({\beta_{lk}^{l}}) = 0$ must have $t_l=0$ at the optimal point, while all other cells can have $t_{l'}>0$.
	Define $f^{*}$ as
	\begin{equation}
	\begin{aligned}
	f^{*} &= \prod_{l'' =1}^{L'} \log_2\left(1+\epsilon+t_{l''}\right)  \prod_{l= L'+1}^L \log_2\left(1+\epsilon+t_l\right).
	\end{aligned}
	\end{equation}
	Suppose $t_{l} = 0 $ for $l = 1,\dots,L'$ and $t_{l} > 0$ otherwise. Therefore, we have
	\begin{equation}
	\begin{aligned}
	f^{*} &=\left(\log_2\left(1+\epsilon\right)\right)^{L'} \prod_{l = L'+1}^L \log_2\left(1+\epsilon+t_l\right),
	\end{aligned}
	\end{equation}
	where $f^{*} = 0 $ for $\epsilon = 0$ and $f^{*}$ is non-zero for any $\epsilon>0$. When solving \eqref{eq:opt_problem1}, the cells with $t_{l}=0$ are effectively removed from the objective function, while the remaining cells will have $t_{l'}>0$ at the optimal point.
\end{proof}

We define $1_{\epsilon}= 1+\epsilon$ to simplify the notation in the remainder of the paper. Note that by using a very small $\epsilon$, the proposed optimization problem achieves scalability without affecting the overall utility. The first constraint in \eqref{eq:opt_problem1} deals with the power control coefficients for the UL data transmission of the users in each cell, and the second constraint is to guarantee a particular SINR $t_l$ to all $K$ users in cell $l$. This constraint guarantees to give the same SINR to every user within a cell, but the SINR value can be different from other cells. Therefore, a cell where all users have poor channels will not prevent the users in other cells from achieving higher SINRs. The GM utility of the per-cell MMF SEs provides proportional fairness between cells. The following optimization problem is the DL counterpart to \eqref{eq:opt_problem1}:  
\begin{equation}\label{eq:opt_problem1DL}
\begin{aligned}
& \underset{\{t_l\},\{\eta_{lk}\}}{\mathrm{maximize}}
& &  \prod_{l=1}^L \log_2\left(1_{\epsilon}+t_l\right)   \\
& {\mathrm{subject~to}}
& &   \eta_{lk} \geq 0,~\forall~l,k, \\
& & &  \sum_{k = 1}^{K}\eta_{lk} \leq 1,~\forall~l, \\
& & & {\mathrm{SINR}}^{\rm{dl}}_{lk} \left(\left\{\eta_{lk}\right\}\right) \geq t_l,~\forall~l,k.\\
\end{aligned}
\end{equation}
The differences from the UL are the SINR expressions being used and the power constraints, which are now reflecting the fact that each BS can distribute its power arbitrarily between its users in the cell. Note that we can use the corresponding SINR expressions for DL data transmission with respect to the channel model assumption provided in \eqref{eq:dl_sinr_corr} and \eqref{eq:dl_sinr}.

Note that the optimization problem formulations for UL and DL that are given in \eqref{eq:opt_problem1} and \eqref{eq:opt_problem1DL} are different and both are important to solve in practice. First of all, the SE expressions are different since the interference comes from different transmitters, which makes the objective functions different. Moreover, the power constraints are different, which makes the feasible sets different. There are individual power constraints for each user in the UL, while in the DL there are sum power constraints for each BS. However, we show that our proposed optimization criterion is applicable for both cases and there is not a significant difference between them in terms of optimization methods for the proposed scheme, but the simulation results in Section \ref{result} show that the actual SEs are behaving very differently, which is why it is important to study both cases.

\subsection{Network-wide max-min fairness}
\label{NW-MMF formulation}
Here, we consider NW-MMF power control for multi-cell massive MIMO network, which has been previously studied in \cite{cheng2017optimal,van2016joint,van2017joint,yang2017massive,xiang2014massive}. Notice that NW-MMF is the ideal utility function in a network where everyone has the same demand for data. It provides equal performance among all the users by prioritizing the user with the weakest channel. The NW-MMF problem for UL data transmission is defined as \cite[Ch.~7]{bjornson2017massive}
\begin{equation}\label{eq:opt_problem2}
\begin{aligned}
& \underset{\{\eta_{lk}\}}{{\mathrm{maximize}}}\quad \underset{l,k}{\mathrm{min}}
& &  {\mathrm{SINR}}_{lk}^{\rm{ul}} \left(\left\{\eta_{lk}\right\}\right)  \\
& {\mathrm{subject~to}}
& &  0 \leq \eta_{lk} \leq 1, \forall~l,k. \\
\end{aligned}
\end{equation}
Note that ${\mathrm{SINR}}_{lk}^{\rm{ul}}\left(\left\{\eta_{lk}\right\}\right)$ in \eqref{eq:opt_problem2} can be defined based on the channel model in which we can replace it with either \eqref{eq:ul_sinr_corr} or \eqref{eq:ul_sinr} for correlated and uncorrelated Rayleigh fading channel model, respectively. The alternative models in Section \ref{otherChannelModels} can also be used.

However, the NW-MMF scheme is not scalable and by increasing the number of cells in the network, we may end up with zero SE for all users---uniform but bad performance for everyone. It happens because the probability of having a user in deep fade due to shadow fading increases and this penalizes the whole network.

\begin{lemma}
	We assume $\{\beta^{l}_{lk}\}$ to be a set of i.i.d. random variables with $\beta^{l}_{lk} \in [0,1]$ and their CDF is denoted by $F(x) \in [0,1]$. In addition, we assume that $F(\epsilon) > 0,\quad\forall \epsilon>0$. In a multi-cell massive MIMO system  
	\begin{equation}
	\underset{l,k}{\mathrm{min}}~\{\beta_{lk}^{l}\} \to 0 ~\textrm{as}~L \to \infty.
	\end{equation}
	Due to this fact and with the SINRs given by \eqref{eq:SINR_general} and $d_{lk}>0$, it follows that	
	\begin{equation} \label{eq:SINR-convergence}
	\underset{l,k}{\mathrm{min}}~{\mathrm{SINR}} _{lk}\left(\left\{\eta_{lk}\right\}\right) \to 0~\textrm{as}~L \to \infty.
	\end{equation}
	Hence, SE$\to0$ at the optimal solution of NW-MMF problem for all the users as $L\to \infty$.
\end{lemma}
\begin{proof}
Define $B  = \underset{l,k}{\mathrm{min}}~\{\beta_{lk}^{l}\} $. We write the probability of $B \le \epsilon$ (where $\epsilon$ is an arbitrary non-negative small number) as  
\begin{equation}
\mathrm{Pr}\left(B \leq \epsilon\right) = 1-\left(1-F\left(\epsilon\right)\right)^{LK}.
\end{equation}
By taking the limit when $L \to \infty$
\begin{equation}
\underset{L \to \infty}{\lim}\mathrm{Pr}\left(B \leq \epsilon \right) = 1.  
\end{equation}
Therefore due to the assumption that $\epsilon$ is an arbitrary non-negative small number,  $B$ converges in distribution to zero when $L \to \infty$.
The SINR in \eqref{eq:SINR_general} is upper bounded by $\frac{a_{lk} \eta_{lk}}{d_{lk}}$. Note that $a_{lk}$ of the worst user goes to zero when $\underset{l,k}{\mathrm{min}}~\{\beta_{lk}^{l}\} \to 0$ in all the considered models. This proves \eqref{eq:SINR-convergence} and implies that the utility function in \eqref{eq:opt_problem2} goes to zero. Therefore, an optimal solution is to give zero SE for all the users.
\end{proof}

The NW-MMF optimization problem for the DL data transmission is similar to \eqref{eq:opt_problem2} but we use the corresponding SINR expression for the DL data transmission by using \eqref{eq:dl_sinr_corr} or \eqref{eq:dl_sinr} to express ${\mathrm{SINR}}^{\rm{dl}}_{lk}\left(\left\{\eta_{lk}\right\}\right)$ depending on the channel model assumption. The NW-MMF optimization problem is written as
\begin{equation}\label{eq:opt_problem2DL}
\begin{aligned}
& \underset{\{\eta_{lk}\}}{{\mathrm{maximize}}}\quad \underset{l,k}{{\mathrm{min}}}
& &  {\mathrm{SINR}}_{lk}^{\rm{dl}} \left(\left\{\eta_{lk}\right\}\right) \\
& {\mathrm{subject~to}}
& & \eta_{lk} \geq 0, \forall~l,k, \\
& & &\sum_{k=1}^{K} \eta_{lk} \leq 1,~\forall l.
\end{aligned}
\end{equation}
This optimization problem also suffers from the same scalability issue as \eqref{eq:opt_problem2}. This is due to the fact that {\it{Lemma}} $1$ applies to the general SINR expression in \eqref{eq:SINR_general}, which covers DL data transmission as well. Note that adding a small $\epsilon>0$ to the objective functions of \eqref{eq:opt_problem2} and \eqref{eq:opt_problem2DL} is not solving the scalability issue of NW-MMF. This is due to the fact that in the case of zero SINR for any user in the network, the optimal solution is to assign $\log_2(1+\epsilon)$ as the SE for all the users. Due to the small value of $\epsilon$, $\log_2(1+\epsilon)$ is approximately equal to zero. Therefore, one can say that the solution still provide zero SE for all the users.

\subsection{Network-wide proportional fairness}
\label{NW-PF}
Next, we consider an alternative network utility function with the product of the SINRs. Maximizing this objective provides NW-PF with respect to the SINRs of the users in the network. It is shown in  \cite[Sec.~7.1]{bjornson2017massive} that this objective is a lower bound on the sum SE of the network, but with greater emphasis on fairness since the utility is zero if any user gets zero SE. 

 We can write the optimization problem for UL data transmission as 
\begin{equation}\label{eq:opt_problem3}
\begin{aligned}
& \underset{\{t_{lk}\},\{\eta_{lk}\}}{\mathrm{maximize}}
& & \prod_{l=1}^L \prod_{k=1}^K t_{lk}   \\
& {\mathrm{subject~to}}
& &  0 \leq \eta_{lk} \leq 1, \forall l,k, \\
& & & {\mathrm{SINR}}^{\rm ul}_{lk} \left(\left\{\eta_{lk}\right\}\right)\geq t_{lk}, \forall l,k,\\
\end{aligned}
\end{equation}
where $t_{lk}$ indicates the effective SINR of user $k$ located at cell $l$. 
The corresponding DL optimization problem is formulated as
\begin{equation}\label{eq:opt_problem3DL}
\begin{aligned}
& \underset{\{t_{lk}\},\{\eta_{lk}\}}{\mathrm{maximize}}
& & \prod_{l=1}^L \prod_{k=1}^K t_{lk}   \\
& {\mathrm{subject~to}}
& &   \eta_{lk} \geq 0,k, \forall l \\
& & & \sum_{k=1}^{K}\eta_{lk} \leq 1,  \forall l \\
& & & {\mathrm{SINR}}^{\rm dl}_{lk} \left(\left\{\eta_{lk}\right\}\right) \geq t_{lk}, \forall  l,k.\\
\end{aligned}
\end{equation}
The difference between this optimization problem and NW-MMF is that this optimization problem deals with each user individually, so there will be large SE differences within a cell. The SINR expressions in \eqref{eq:opt_problem3} and \eqref{eq:opt_problem3DL} can be replaced with the corresponding expression based on the data transmission direction and channel model assumption.

However, similar to NW-MMF, this scheme is not scalable and by increasing the number of cells in the network, we may end up with zero SE for all users. This happens due to the fact that increasing the number of cells in the network or increasing the number of users, increases the probability of having a user in deep fade due to shadow fading which corresponds to $\beta_{lk}^{l} \to 0$ for some $l,k$. Note that  ${\mathrm{SINR}}^{\rm ul/dl}_{lk} \to 0$ if $\beta_{lk}^{l} \to 0$ implies that $\prod\limits_{l=1}^L \prod\limits_{k=1}^K t_{lk} \to 0$. In other words, if one user has a zero SINR then the product of the SINRs becomes zero as well. Hence, giving zero SINR to all users will be an optimal solution to the NW-PF problem. One could think of using a similar trick as provided in our proposed method for solving this issue by introducing a new control parameter $\epsilon > 0$ and  replacing the objective functions in \eqref{eq:opt_problem3} and \eqref{eq:opt_problem3DL} by $ \underset{\{t_{lk}\},\{\eta_{lk}\}}{\mathrm{maximize}} \prod\limits_{l=1}^L \prod\limits_{k=1}^K \left(t_{lk} +\epsilon\right)$. However, after adding the $\epsilon$, the problem becomes much harder to solve and cannot be tackled using the approach described in the next section. The complexity is basically the same as maximizing sum SE, which is known to be an NP-hard problem. Therefore, the same approach does not solve the scalability issue of NW-PF problem. Hence, we can see that our proposed method is a practical approach that provides scalable power control for multi-cell massive MIMO systems.
Note that, in the case of NW-MMF and NW-PF, having a user with zero SINR in the network will result in zero SINR for all the users in the whole network. This happens due to the optimization objectives in the problem formulations \eqref{eq:opt_problem2}, \eqref{eq:opt_problem2DL}, \eqref{eq:opt_problem3}, and \eqref{eq:opt_problem3DL}. In our proposed optimization formulation, the SINR constraint in \eqref{eq:opt_problem1} is given as ${\mathrm{SINR}}^{\rm ul}_{lk} \left(\left\{\eta_{lk}\right\}\right) \geq t_l, \forall~l,k,$ which corresponds to max-min fairness within each cell $l$. If we have one user that can only get zero SINR, then that will affect the serving cell of this user and lead to zero SINR to everyone located in that cell. However, this does not apply to the users in other cells, which operate as if the cell with the zero-SINR user does not exist.

\section{Solutions to the Proposed Problems}
\label{solution}
In this section, we provide solutions to the optimization problems introduced in Section \ref{probelms}. First, we solve the proposed GM per-cell MMF power control for the UL data transmission given in \eqref{eq:opt_problem1}. We can rewrite the optimization problem as
\begin{equation}\label{eq:opt_problem_re}
\begin{aligned}
& \underset{\{t_l\},\{\eta_{lk}\}}{\mathrm{maximize}}
& &  \sum_{l=1}^L \log \left( \log_2\left(1_{\epsilon}+t_l\right)\right)   \\
& {\mathrm{subject~to}}
& &  0 \leq \eta_{lk} \leq 1, \forall~l,k, \\
& & & {\mathrm{SINR}}^{\rm ul}_{lk}\left(\left\{\eta_{lk}\right\}\right) \geq t_l, \forall~l,k,\\
\end{aligned}
\end{equation}
since the logarithm is a monotonically increasing function. The problems are equivalent in the sense that solving them provides the same optimal solution. We will now prove that \eqref{eq:opt_problem_re} can be rewritten as a convex problem. By the change of variables
\begin{equation}
\begin{aligned}
t_l &= e^{\bar{t}_l},\quad\quad\eta_{lk} & = e^{\bar{\eta}_{lk}},
\end{aligned}
\end{equation}
we obtain the following  equivalent problem:
\begin{equation}\label{eq:opt_problem_re2}
\begin{aligned}
& \underset{\{\bar{t}_l\},\{\bar{\eta}_{lk}\}}{\mathrm{maximize}}
& &  \sum_{l=1}^L \log\left(\log_2\left(1_{\epsilon}+e^{\bar{t}_l}\right)\right)   \\
& {\mathrm{subject~to}}
& & e^{\bar{\eta}_{lk}} \leq 1, \forall~l,k \\
& & & \hspace{-1.3cm}\frac{a_{lk}e^{\bar{\eta}_{lk}}}{\sum\limits_{l'=1}^{L}\sum\limits_{k' = 1}^{K} b^{l'k'}_{lk}  e^{\bar{\eta}_{l'k'}}+\sum\limits_{l'\in \mathcal{P}_{l}\backslash \{ l \}}^{}c^{l'}_{lk}e^{\bar{\eta}_{l'k}}+d_{lk}} \geq e^{\bar{t}_l}, \forall~l,k.\\
\end{aligned}
\end{equation}
Note that, we used the general SINR expression of \eqref{eq:SINR_general} in the last constraint of \eqref{eq:opt_problem_re2}, thus the problem formulation applies to both correlated and uncorrelated fading or any of the other cases mentioned in Section \ref{otherChannelModels}. We observe that the constraints can be rewritten as
\begin{equation}
\begin{aligned}
& \sum\limits_{l'=1}^{L}\sum\limits_{k'=1}^{K} b^{l'k'}_{lk} e^{\bar{\eta}_{l'k'}+\bar{t}_l-\bar{\eta}_{lk}} \\
&+\sum\limits_{l'\in \mathcal{P}_{l}\backslash \{ l \}}^{}c^{l'}_{lk}e^{\bar{\eta}_{l'k}+\bar{t}_l-\bar{\eta}_{lk}}+d_{lk}e^{\bar{t}_l-\bar{\eta}_{lk}}\leq a_{lk}.
\end{aligned}
\end{equation}
After taking the logarithm of both sides, we have a log-sum-exponential function, which is a convex function less than or equal to a constant. This is a convex constraint. In addition, we take the logarithm of both sides in the other constraints i.e., $e^{\bar{\eta}_{lk}} \leq 1, \forall~l,k $, which becomes ${\bar{\eta}_{lk}} \leq 0, \forall~l,k,$ and it is a convex constraint as well. Hence, finally we have the equivalent problem as given in \eqref{eq:opt_problem_re2_fin}.
\begin{figure*}
\begin{equation}\label{eq:opt_problem_re2_fin}
\begin{aligned}
& \underset{\{\bar{t}_l\},\{\bar{\eta}_{lk}\}}{\mathrm{maximize}}
& &  \sum_{l=1}^L \log\left(\log_2\left(1_{\epsilon}+e^{\bar{t}_l}\right)\right)   \\
& {\mathrm{subject~to}}
& & \bar{\eta}_{lk} \leq 0, \forall~l,k \\
& & &\log\left( \sum\limits_{l'=1}^{L}\sum\limits_{k'=1}^{K} b^{l'k'}_{lk} e^{\bar{\eta}_{l'k'}+\bar{t}_l-\bar{\eta}_{lk}} +\sum\limits_{l'\in \mathcal{P}_{l}\backslash \{ l \}}^{}c^{l'}_{lk}e^{\bar{\eta}_{l'k}+\bar{t}_l-\bar{\eta}_{lk}}+ d_{lk}e^{\bar{t}_l-\bar{\eta}_{lk}}+\right)\leq \log\left(a_{lk}\right), \forall~l,k.\\
\end{aligned}
\end{equation}
\vspace*{-0.5cm}
\end{figure*}

Therefore the only remaining concern is whether the objective function in \eqref{eq:opt_problem_re2} is concave or not. We provide the following theorem that shows that the objective function is a concave function and thus \eqref{eq:opt_problem_re2} is a convex problem. 
\begin{lemma}
	The function  $ \log \left(\log_2 \left( 1_{\epsilon}+ e^{x}\right)\right)$ is a concave function with respect to $x,\forall x$.
\end{lemma}
\begin{proof}
	The proof is provided in Appendix \ref{proof1}.
\end{proof}
As the objective is a sum of concave functions, it is also jointly concave. Hence, we have shown that \eqref{eq:opt_problem1} is a convex problem and it follows that every stationary point is also a global optimum solution. Note that solving the optimization problem for the DL case follows the same steps as UL, hence the detailed derivation is omitted to avoid repetition. For the DL data transmission, we use \eqref{eq:opt_problem1DL} and apply a change of variables similar to the UL case. This provides the equivalent convex problem provided in \eqref{eq:opt_problemDL_re2_fin}.
\begin{figure*}
\begin{equation}\label{eq:opt_problemDL_re2_fin}
\begin{aligned}
& \underset{\{\bar{t}_l\},\{\bar{\eta}_{lk}\}}{\mathrm{maximize}}
& &  \sum_{l=1}^L \log\left(\log_2\left(1_{\epsilon}+e^{\bar{t}_l}\right)\right)   \\
& {\mathrm{subject~to}}
& & \sum_{k = 1}^{K}\bar{\eta}_{lk} \leq 0,~\forall~l, \\
& & &\log\left( \sum\limits_{l'=1}^{L}\sum\limits_{k'=1}^{K} b^{l'k'}_{lk} e^{\bar{\eta}_{l'k'}+\bar{t}_l-\bar{\eta}_{lk}} +\sum\limits_{l'\in \mathcal{P}_{l}\backslash \{ l \}}^{}c^{l'}_{lk}e^{\bar{\eta}_{l'k}+\bar{t}_l-\bar{\eta}_{lk}}+ d_{lk}e^{\bar{t}_l-\bar{\eta}_{lk}}\right)\leq \log\left(a_{lk}\right), \forall~l,k.\\
\end{aligned}
\end{equation}
\hrulefill
\end{figure*}
We showed that both \eqref{eq:opt_problem_re2_fin} and \eqref{eq:opt_problemDL_re2_fin} are convex. Therefore any algorithm that converges to a stationary point can be applied to solve these problems. In the simulation part, we use the interior point algorithm within the \texttt{fmincon} solver in MATLAB. The computational complexity of this approach is of the order of $\mathcal{O} \left(\max\{(KL+L)^3, (KL+L)^2(2KL),F \}\right)$, where $F$ is the cost of calculating the first and second derivatives of the objective and constraint functions \cite{Boyd}. Note that the complexity is calculated as the number of operations required in each iteration of the interior-point method.

\subsection{Per-cell MMF approximate solution}
Depending on the computational resources, the proposed GM per-cell MMF power control algorithm can either be implemented in real-time or used to benchmark and design heuristic power control solutions. The authors in \cite{redbook} proposed an approximate solution for the per-cell max-min power control problem with uncorrelated Rayleigh fading. In this subsection, we generalize it to correlated Rayleigh fading channels and the other cases mentioned in Section \ref{otherChannelModels}.  
For the UL data transmission, we first ignore the coherent interference part (i.e., $c^{l'}_{lk}$) from the denominator of \eqref{eq:SINR_general} and calculate the data power coefficients as
\begin{equation}
\begin{aligned}
\hat{\eta}_{lk} = \frac{\underset{k'}{\mathrm{min}}\{a_{lk'}\}}{a_{lk}}, \forall l,k.
\end{aligned}
\end{equation}
We use the power coefficients $\hat{\eta}_{lk}$ to calculate the exact SINRs using  \eqref{eq:SINR_general} and denote it as $\widehat{\mathrm{SINR}}^{\rm ul}_{lk}$. Then following the same approach as in \cite{redbook}, the approximate SINR of user $k$ in cell $l$ can be defined as
\begin{equation}
\begin{aligned}
{\mathrm{SINR}}^{\rm ul}_{lk} = \underset{k'}\min \left\{\frac{\widehat{\mathrm{SINR}}^{\rm ul}_{lk'}}{\hat{\eta}_{lk'}}\right\}.
\end{aligned}
\end{equation}
The power coefficients of DL data transmission of user $k$ in cell $l$ when we neglect the coherent interference is defined as \cite{redbook}
\begin{equation}
\begin{aligned}
\hat{\eta}_{lk} = \frac{d_{lk}+\sum\limits_{l' = 1}^{L} b^{l'}_{lk}}{a_{lk} \sum\limits_{k' = 1}^{K} \frac{d_{lk}+\sum\limits_{l' = 1}^{L} b^{l'}_{lk'}}{a_{lk'}}}.
\end{aligned}
\end{equation}
Note that we write $b^{l'}_{lk}= b^{l'k'}_{lk}$ due to the fact that in the DL $\{b^{l'k'}_{lk}\}$ do not depend on $k'$. Similar to the UL, first, we use the power coefficients $\hat{\eta}_{lk}$ to calculate the exact DL SINRs by utilizing  \eqref{eq:SINR_general} which is denoted as $\widehat{\mathrm{SINR}}^{\rm dl}_{lk}$. Then, the per-cell MMF approximate SINR of user $k$ in cell $l$ for DL is defined as \cite{redbook}
\begin{equation}
\begin{aligned}
{\mathrm{SINR}}^{\rm dl}_{lk} = \frac{1}{\sum\limits_{k' = 1}^{K} \frac{\hat{\eta}_{lk'}}{\widehat{\mathrm{SINR}}^{\rm dl}_{lk'}}}.
\end{aligned}
\end{equation}
We define and summarize the corresponding coefficients for both UL and DL data transmission for uncorrelated and correlated channel model in Table \ref{coeff}.
\begin{figure*}[t]
	\begin{equation}\label{eq:general_fading}
	\begin{aligned}
	\mathrm{SINR}^{\rm ul}_{lk} \left(\left\{\eta_{lk}\right\}\right)  = \frac{\eta_{lk}\rho_{\rm ul} \left|\mathbb{E}\left\{ \mathbf{v}^{\rm H}_{lk} \mathbf{h}^{l}_{lk}\right\}\right|^2}{\sum\limits_{l' = 1}^{L}\sum\limits_{k' = 1}^{K} \eta_{l'k'} \rho_{\rm ul} \mathbb{E}\left\{\left|\mathbf{v}^{\rm H}_{lk}\mathbf{h}^{l}_{l'k'}\right|^2\right\} - \eta_{lk}\rho_{\rm ul}\left|\mathbb{E}\left\{\mathbf{v}^{\rm H}_{lk}\mathbf{h}^{l}_{lk}\right\}\right|^2 + \sigma^2\mathbb{E}\left\{\left\lVert \mathbf{v}_{lk}\right\lVert^2 \right\}}.
	\end{aligned}
	\end{equation}
	\begin{equation}
	\begin{aligned}\label{eq:general_fading_re}
	\mathrm{SINR}^{\rm ul}_{lk} \left(\left\{\eta_{lk}\right\}\right) = \frac{\eta_{lk}\rho_{\rm ul} \left|\mathbb{E}\left\{ \left(\hat{\mathbf{h}}^{l}_{lk}\right)^{\rm H} \mathbf{h}^{l}_{lk}\right\}\right|^2}{\sum\limits_{l' = 1}^{L}\sum\limits_{k' = 1}^{K} \eta_{l'k'} \rho_{\rm ul} \mathrm{var}\left\{\left(\hat{\mathbf{h}}^{l}_{lk}\right)^{\rm H}\mathbf{h}^{l}_{l'k'}\right\}+ \sum\limits_{l'\in \mathcal{P}_{l}\backslash \{l \}} \eta_{lk}\rho_{\rm ul}\left|\mathbb{E}\left\{\left(\hat{\mathbf{h}}^{l}_{lk}\right)^{\rm H}\mathbf{h}^{l}_{l'k}\right\}\right|^2 + \sigma^2\mathbb{E}\left\{\left\lVert \hat{\mathbf{h}}_{lk}\right\lVert^2 \right\}}.
	\end{aligned}
	\end{equation}
	\begin{equation}
	\begin{aligned}\label{eq:dl_general_1}
	\mathrm{SINR}^{\rm dl}_{lk} \left(\left\{\eta_{lk}\right\}\right) = \frac{\eta_{lk}\rho_{\rm dl} \left|\mathbb{E}\left\{ \mathbf{w}^{\rm H}_{lk} \mathbf{h}^{l}_{lk}\right\}\right|^2}{\sum\limits_{l' = 1}^{L}\sum\limits_{k' = 1}^{K} \eta_{l'k'} \rho_{\rm dl} \mathbb{E}\left\{\left|\mathbf{w}^{\rm H}_{l'k'}\mathbf{h}^{l'}_{lk}\right|^2\right\} - \eta_{lk}\rho_{\rm dl}\left|\mathbb{E}\left\{\mathbf{w}^{\rm H}_{lk}\mathbf{h}^{l}_{lk}\right\}\right|^2 + \sigma^2}.
	\end{aligned}
	\end{equation}
	\begin{equation}\label{eq:dl_general_1_re}
	\begin{aligned}
	\mathrm{SINR}^{\rm dl}_{lk} \left(\left\{\eta_{lk}\right\}\right) = \frac{\eta_{lk}\rho_{\rm dl} \left|\mathbb{E}\left\{ \frac{\left(\hat{\mathbf{h}}^{l}_{lk}\right)^{\rm H}\mathbf{h}^{l}_{lk}}{\sqrt{\mathbb{E}\left\{\lVert\hat{\mathbf{h}}^{l }_{lk}\lVert^2\right\}}} \right\}\right|^2}{\sum\limits_{l' = 1}^{L}\sum\limits_{k' = 1}^{K} \eta_{l'k'} \rho_{\rm dl} \mathrm{var}\left\{\frac{\left(\hat{\mathbf{h}}^{l'}_{l'k'}\right)^{\rm H}\mathbf{h}^{l'}_{lk }}{\sqrt{\mathbb{E}\left\{\lVert\hat{\mathbf{h}}^{l'}_{l'k'}\lVert^2\right\}}}\right\}+ \sum\limits_{ l'\in \mathcal{P}_{l}\backslash \{l \}} \eta_{lk}\rho_{\rm dl}\left|\mathbb{E}\left\{\frac{\left(\hat{\mathbf{h}}^{l'}_{l'k}\right)^{\rm H}\mathbf{h}^{l'}_{lk}}{\sqrt{\mathbb{E}\left\{\lVert\hat{\mathbf{h}}^{l'}_{l'k}{\lVert}^2\right\}}}\right\}\right|^2 + \sigma^2}.
	\end{aligned}
	\end{equation}
	\hrulefill
\end{figure*}
\subsection{Solution approach for NW-MMF and NW-PF}
To solve the NW-MMF optimization problems given in \eqref{eq:opt_problem2} and \eqref{eq:opt_problem2DL} one can write these problems on epigraph form and solve linear feasibility optimization problems using the bisection algorithm. The details of the bisection algorithm can be found in \cite[Ch.~7]{bjornson2017massive}. Note that solving the linear feasibility optimization problems using the bisection algorithm is sensitive to the optimality criteria of bisection algorithm. Therefore, to guarantee the actual MMF optimal solution one should select the minimum SE among all users as the optimal solution for MMF problem.

The NW-PF problem for both UL and DL data transmission, given in \eqref{eq:opt_problem3} and \eqref{eq:opt_problem3DL}, are geometric programming problems \cite{boyd2007tutorial,gpPowercont}. The detailed proof is provided in \cite[Th.~7.2]{bjornson2017massive}. These optimization problems can be solved efficiently by using standard convex optimization solvers, for example, we used CVX \cite{grant2008cvx} in the simulation part.

\section{Other channel models}
\label{otherChannelModels}
In this part, we introduce other channel models for which the SINR expression follows the same structure as \eqref{eq:SINR_general}. The main intention is to show that the solutions provided for the power control problems in the previous sections are not limited to correlated/uncorrelated Rayleigh fading channel models. The problem formulations and the proposed solutions are covering more general cases even though we just consider correlated and uncorrelated Rayleigh fading channels for the numerical evaluation, for brevity.

Assuming the same block fading model but a general channel fading model and an arbitrary linear detection vector $\mathbf{v}_{lk} \in \mathbb{C}^M$, the UL SINR expression is given in \eqref{eq:general_fading} \cite[Th.~4.4]{bjornson2017massive}. We can rewrite \eqref{eq:general_fading} as given in \eqref{eq:general_fading_re} for MR processing at the BSs i.e., $\mathbf{v}_{lk} =  \hat{\mathbf{h}}^{l}_{lk}$.
For this case, the corresponding parameters to be used in the general SINR expression \eqref{eq:general_fading} are defined as 
\begin{equation*}
\begin{aligned}
a_{lk} &= \rho_{\rm ul} \left|\mathbb{E}\left\{ \left(\hat{\mathbf{h}}^{l}_{lk}\right)^{\rm H} \mathbf{h}^{l}_{lk}\right\}\right|^2,\\
b^{l'k'}_{lk} &= \rho_{\rm ul} \mathrm{var}\left\{\left(\hat{\mathbf{h}}^{l}_{lk}\right)^{\rm H}\mathbf{h}^{l}_{l'k'}\right\},\\
c^{l'}_{lk} &= \rho_{\rm ul}\left|\mathbb{E}\left\{\left(\hat{\mathbf{h}}^{l}_{lk}\right)^{\rm H}\mathbf{h}^{l}_{l'k}\right\}\right|^2,\\
d_{lk} &= \sigma^2\mathbb{E}\left\{\left\lVert \hat{\mathbf{h}}_{lk}\right\lVert^2 \right\}.
\end{aligned}
\end{equation*}
In the same case, the DL data transmission with arbitrary linear detection vectors $\mathbf{w}_{lk} \in \mathbb{C}^M$ leads to the general SINR expression that is provided in \eqref{eq:dl_general_1}. For the case of MR processing at the BSs with arbitrary zero mean channel fading model, we define $\mathbf{w}_{lk} = \frac{\hat{\mathbf{h}}^{l}_{l,k}}{\sqrt{\mathbb{E}\left\{\lVert\hat{\mathbf{h}}^{l}_{lk}\lVert^2\right\}}}$ and plug it into \eqref{eq:dl_general_1}. It gives us the SINR expression which has same structure as \eqref{eq:SINR_general} which is given in \eqref{eq:dl_general_1_re}. Similar to the UL case, we define the corresponding parameters for the general SINR expression as follows
\begin{equation}
\begin{aligned}
a_{lk} &= \rho_{\rm dl} \left|\mathbb{E}\left\{ \frac{\left(\hat{\mathbf{h}}^{l}_{lk}\right)^{\rm H}\mathbf{h}^{l}_{lk}}{\sqrt{\mathbb{E}\left\{\lVert\hat{\mathbf{h}}^{l }_{lk}\lVert^2\right\}}} \right\}\right|^2,\\
b^{l'k'}_{lk} &= \rho_{\rm dl} \mathrm{var}\left\{\frac{\left(\hat{\mathbf{h}}^{l'}_{l'k'}\right)^{\rm H}\mathbf{h}^{l'}_{lk }}{\sqrt{\mathbb{E}\left\{\lVert\hat{\mathbf{h}}^{l'}_{l'k'}\lVert^2\right\}}}\right\}, \\
c^{l'}_{lk} &= \rho_{\rm dl}\left|\mathbb{E}\left\{\frac{\left(\hat{\mathbf{h}}^{l'}_{l'k}\right)^{\rm H}\mathbf{h}^{l'}_{lk}}{\sqrt{\mathbb{E}\left\{\lVert\hat{\mathbf{h}}^{l'}_{l'k}{\lVert}^2\right\}}}\right\}\right|^2,~~~ d_{lk} = \sigma^2.
\end{aligned}
\end{equation} 
Note that these expressions for general channel fading model can be applied for Rician and Nakagami fading channels, and many other examples. 
Another example is the case of line-of-sight communication in which we assume an arbitrary deterministic channel vector denoted as $\bar{\mathbf{h}}^{l}_{lk}$. In addition, it is assumed that each BS utilizes an arbitrary precoder/combiner. 

For the case of UL data transmission $\mathbf{v}_{lk}$ indicates the combining vector. The SINR expression is given in \eqref{eq:ul_los} provided at the top of this page.
\begin{figure*}
\begin{equation}
\begin{aligned}\label{eq:ul_los}
\mathrm{SINR}^{\rm ul}_{lk} \left(\left\{\eta_{lk}\right\}\right) = \frac{\eta_{lk}\rho_{\rm ul} \left|\mathbf{v}^{\rm H}_{lk} \bar{\mathbf{h}}^{l}_{lk}\right|^2}{\sum\limits_{\substack{k' = 1\\ k'\neq k}}^{K} \eta_{lk'} \rho_{\rm ul} \left|\mathbf{v}^{\rm H}_{lk}\bar{\mathbf{h}}^{l}_{lk'}\right|^2 + \sum\limits_{\substack{l'= 1\\ l'\neq l}}^{L}\sum\limits_{k' = 1}^{K} \eta_{l'k'} \rho_{\rm ul} \left|\mathbf{v}^{\rm H}_{lk}\bar{\mathbf{h}}^{l}_{l'k'}\right|^2 + \sigma^2\left\lVert \mathbf{v}_{lk}\right\lVert^2 }.
\end{aligned}
\end{equation}
\hrulefill
\begin{equation}\label{eq:dl_los}
\begin{aligned}
\mathrm{SINR}^{\rm dl}_{lk} \left(\left\{\eta_{lk}\right\}\right) = \frac{\eta_{lk}\rho_{\rm dl} \left| \mathbf{w}^{\rm H}_{lk} \bar{\mathbf{h}}^{l}_{lk}\right|^2}{\sum\limits_{\substack{k' = 1 \\ k'\neq k}}^{K} \eta_{lk} \rho_{\rm dl} \left|\mathbf{w}^{\rm H}_{l'k'}\bar{\mathbf{h}}^{l}_{lk}\right|^2 + \sum\limits_{\substack{l'= 1\\ l'\neq l}}^{L}\sum\limits_{k'= 1}^{K} \eta_{l'k'} \rho_{\rm dl} \left|\mathbf{w}^{\rm H}_{l'k'}\bar{\mathbf{h}}^{l'}_{lk}\right|^2 + \sigma^2}.
\end{aligned}
\end{equation}
\hrulefill
\end{figure*}
In the DL data transmission, $\mathbf{w}_{lk}$ is the precoder vector of BS $l$ to user $k$ and the SINR is defined as it is written in \eqref{eq:dl_los} which is provided in the top of this page.
The SINR expression for both UL and DL communication with line-of-sight channel model provided in \eqref{eq:ul_los} and \eqref{eq:dl_los}, respectively, follow the same structure as \eqref{eq:SINR_general}. In which for the UL, we have 
\begin{equation*}
\begin{aligned}
a_{lk} &= \rho_{\rm ul} \left|\mathbf{v}^{\rm H}_{lk} \bar{\mathbf{h}}^{l}_{lk}\right|^2,\\
b^{l'k'}_{lk} &= \begin{cases}
\rho_{\rm ul} \left|\mathbf{v}^{\rm H}_{lk}\bar{\mathbf{h}}^{l}_{l'k'}\right|^2,& (l',k') \neq (l,k)\\
0,\quad &(l',k') = (l,k),
\end{cases}\\
c^{l'k'}_{lk} & = 0,\quad\quad \quad\quad \quad\quad d_{lk} = \sigma^2\left\lVert \mathbf{v}_{lk}\right\lVert^2.
\end{aligned}
\end{equation*}

In addition, for the DL data transmission case, the parameters are 
\begin{equation*}
\begin{aligned}
a_{lk} &= \rho_{\rm dl} \left| \mathbf{w}^{\rm H}_{lk} \bar{\mathbf{h}}^{l}_{lk}\right|^2,\\
b^{l'k'}_{lk} &= \begin{cases}
\rho_{\rm dl} \left|\mathbf{w}^{\rm H}_{l'k'}\bar{\mathbf{h}}^{l'}_{lk}\right|^2, & (l',k') \neq (l,k)\\
0,\quad &(l',k') = (l,k),
\end{cases}\\
c^{l'k'}_{lk} &= 0,~\quad\quad\quad\quad\quad\quad d_{lk} = \sigma^2.
\end{aligned}
\end{equation*}
Therefore, we can apply our proposed power control scheme for this case as well.

\section{Numerical Analysis}
\label{result}
In this section, we provide a numerical comparison of the four power control algorithms provided in Section \ref{probelms}. We consider a multi-cell massive MIMO setup consisting of $16$ cells. The network uses wrap-around to avoid edge effects. We assume a square grid layout where each square has a BS in the center and all of the BSs are located in a $1$\,km$^2$ area. Furthermore, each BS serves $K = 5$ users that are randomly distributed with uniform distribution in the coverage area of the BS. In addition, it is assumed that each BS is equipped with $M = 100$ antennas and $\epsilon = 0.001$. 
In addition, to investigate the effect of pilot contamination on the performance of power control algorithms, we consider three different pilot reuse factors $f=\{1,2,4\}$ in the numerical analysis.
The bandwidth is $20\,$MHz and each coherence block contains $\tau_c = 200$ symbols. The large-scale fading coefficients are modeled as \cite{bjornson2017massive}
\begin{equation}
\beta^{l}_{l'k} \left[{\rm dB}\right] = -35 - 36.7\log_{10}\left(d^{l}_{l'k}/1\,\rm{m}\right) +F^{l}_{l'k},
\end{equation}
where $d^{l}_{l'k} $ is the distance between user $k$ located in cell $l'$ to BS $l$. In addition, $F^{l}_{l'k}$ is shadow fading generated from a log-normal distribution with standard deviation of 7$\,{\rm dB}\,$.

Note that in the implementation, whenever needed, we regenerated the shadow fading realization to guarantee that the home BS has the largest large-scale fading towards its serving user. It means that $\beta^{l}_{lk}$ is the largest among all $\beta^{l}_{l'k},~~l'  = 1,\dots,L$. The noise variance is set to~$-94\,$dBm and the maximum transmit power of the users is~$200\,$mW for UL data transmission (Note that with the current maximum power settings for UL data transmission the median SNRs for the
cell-edge users is relatively low: around $-3\,$dB). The maximum transmit power of the BS is selected to be $40\,$W. The results are qualitatively the same if we change the power budgets, as is shown in Appendix \ref{budgetVSrate}. The simulations consider $1000$ realizations, where the users are dropped uniformly and randomly in each cell.  

It is assumed that each BS has a horizontal uniform linear array with half-wavelength antenna spacing. Furthermore, the spatial correlation matrix from user $k$ in cell $l$ to the BS $l'$ is modeled using the approximate Gaussian local scattering model provide in \cite[Ch.2.6]{bjornson2017massive} 
\begin{equation}
\left[\mathbf{R}^{l}_{l'k}\right]_{m,n}  =\beta^{l}_{l'k} e^{\pi j (m-n)\sin (\varphi^{l}_{l'k})} e^{- \frac{\sigma^2_{\varphi}}{2} \pi (m-n)\cos (\varphi)},
\end{equation}
where $\varphi^{l}_{l'k}$ is the nominal angle of arrival (AoA) from user $k$ in cell $l'$ to the BS $l$. It is also assumed that the multipath components of a cluster have Gaussian distributed AoA around nominal AoA with a standard deviation of $10^{\textdegree}$ degree in the simulations.

Figs.\ref{fig:corr_ul_f1}-\ref{fig:corr_ul_f4}, plot the cumulative distribution function (CDF) of the SE of all the users for UL data transmission with pilot reuse factor $1$, $2$ and $4$, respectively. In these figures, it can be seen that NW-PF offers higher SE than the proposed GM of per-cell MMF SE for most users but not in the lower tail which is the important part for delivering fairness and uniform performance. 
\begin{figure*}[t]
\centering
\begin{subfigure}{0.32\linewidth}
	\includegraphics[width=\linewidth]{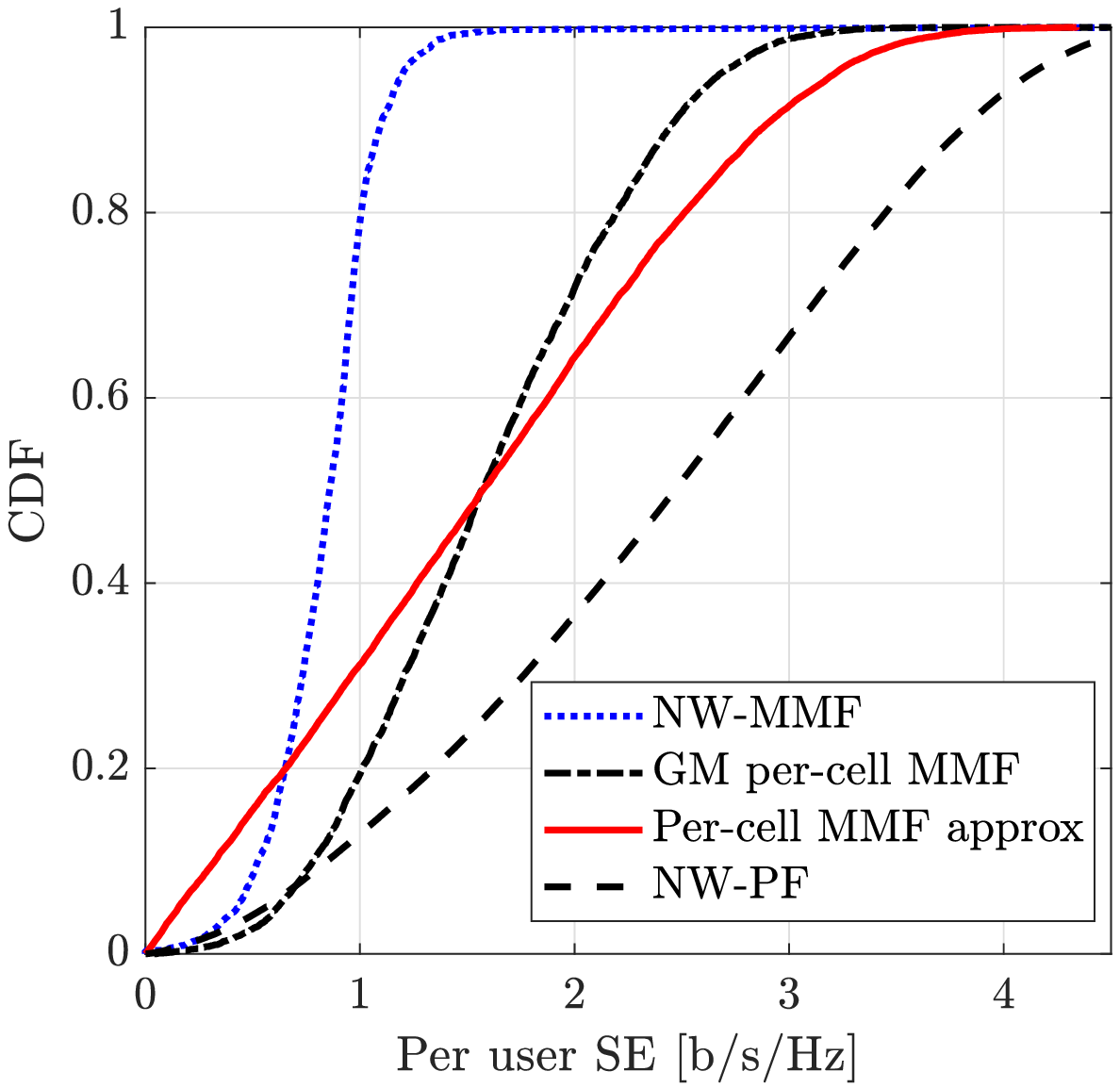}
	\caption{$f=1$.}
	\label{fig:corr_ul_f1}
 \end{subfigure}
 \begin{subfigure}{0.32\linewidth}
	\includegraphics[width=\linewidth]{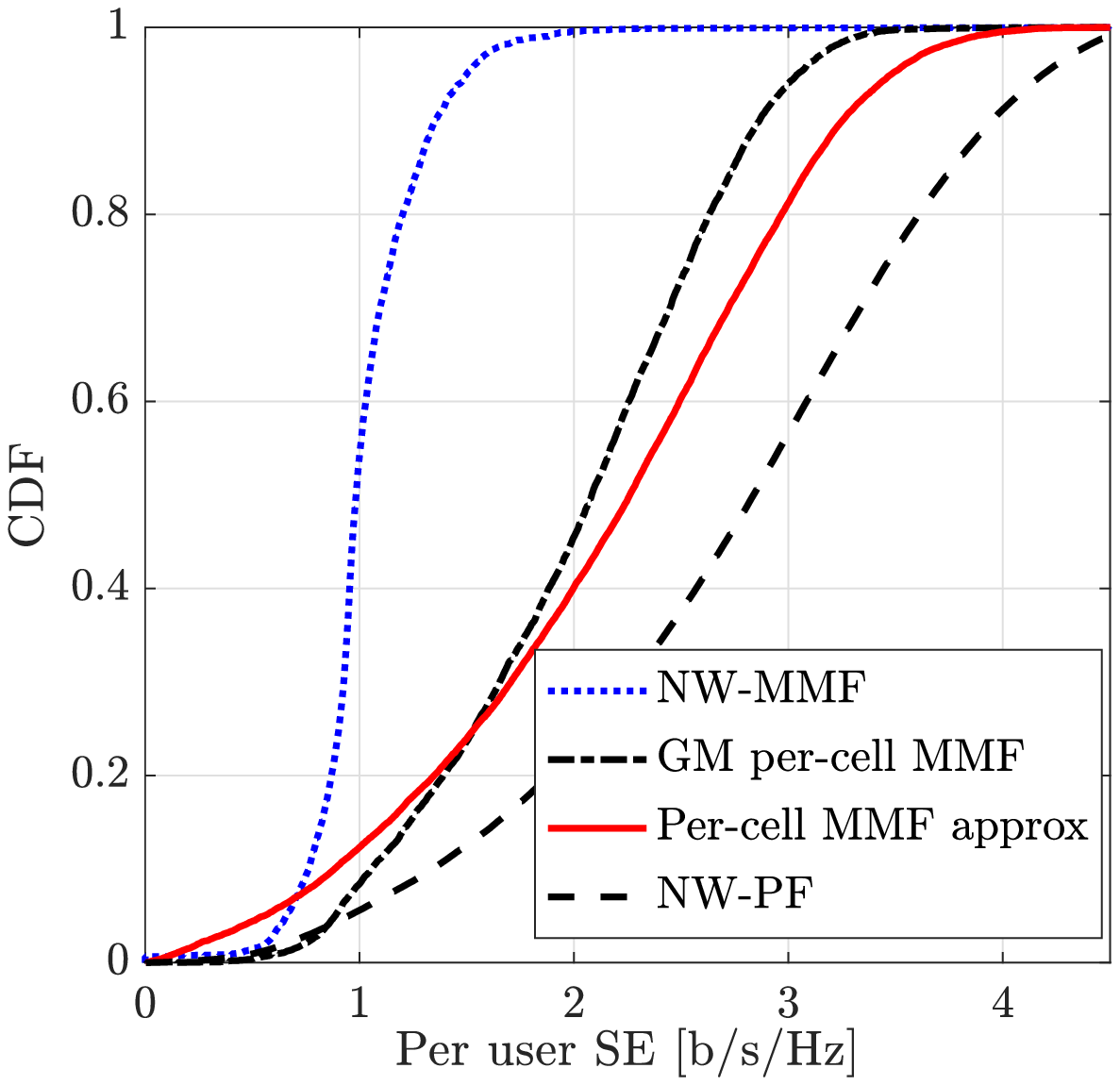}
	\caption{$f=2$.}
	\label{fig:corr_ul_f2}
 \end{subfigure}
  \begin{subfigure}{0.32\linewidth}
	\includegraphics[width=\linewidth]{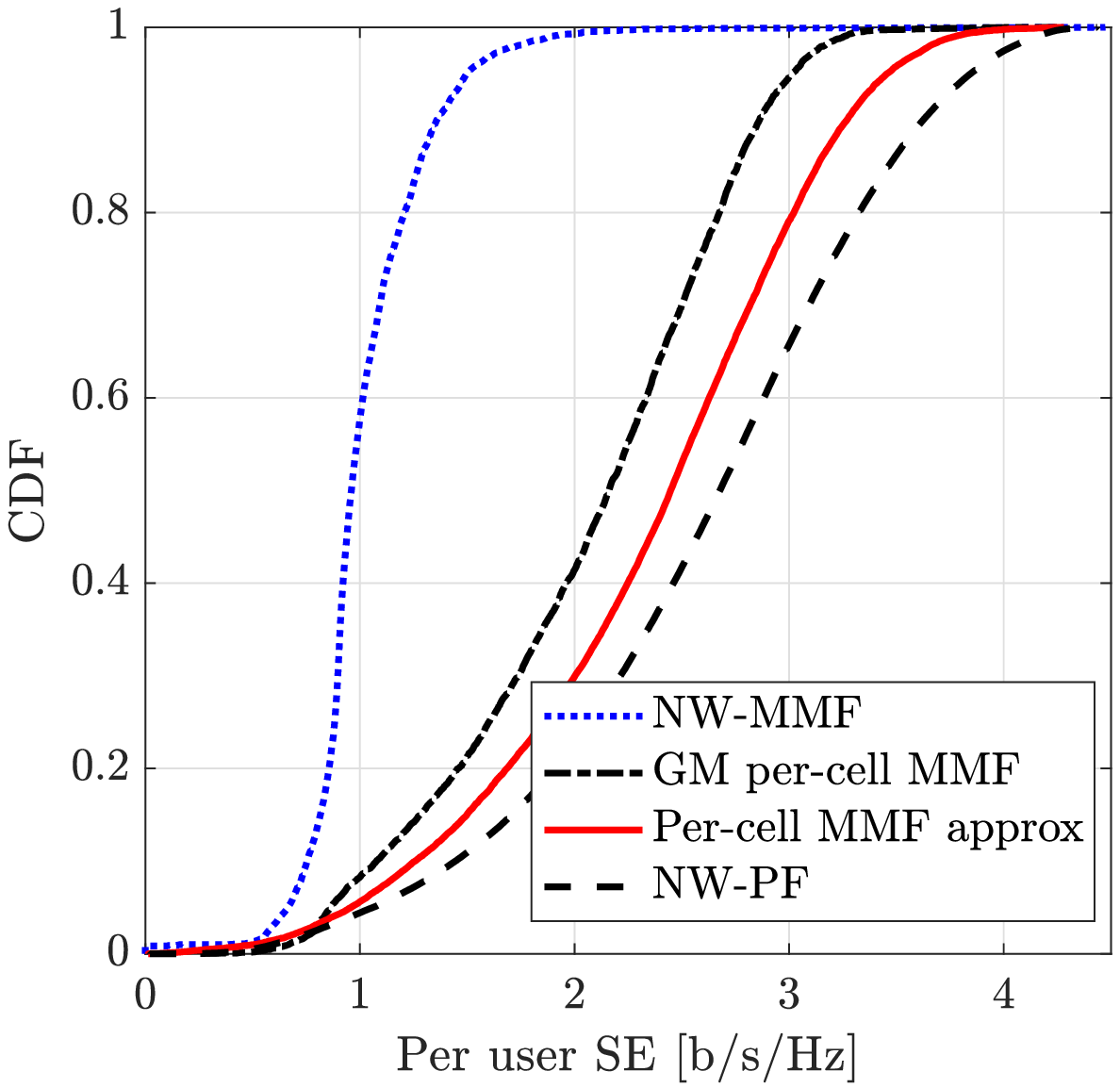}
	\caption{$f=4$.}
	\label{fig:corr_ul_f4}
\end{subfigure}
	\caption{SE of CU $k$ for UL data transmission with correlated Rayleigh fading channel with different pilot reuse factors}
\label{fig:corr_ul}
\end{figure*}

In the case of pilot reuse factor $1$, $7\%$ weakest users get higher SE when using the proposed GM per-cell MMF. For the reuse factor $2$ and $4$, less than $5\%$ and $2\%$ of the weakest users get higher SE when using the proposed GM per-cell MMF SE, respectively. Fig. \ref{fig:95likelyULbar} shows a bar diagram of the 98\%-likely per user SE to clarify the mentioned observations. In addition, our proposed algorithm mostly provides  higher SE in comparison with per-cell MMF approximation for pilot reuse factor $1$ and $2$. However, for pilot reuse factor $4$, in which we have a lower pilot contamination effect, the approximate solution has higher per user SE. The heuristic approach verifies that the approximate solution is close to the optimal solution provided by our proposed method and it can be used as a low complexity approximate alternative solution.

\begin{figure}[htb!]
	\centering
	\includegraphics[width=.8\columnwidth]{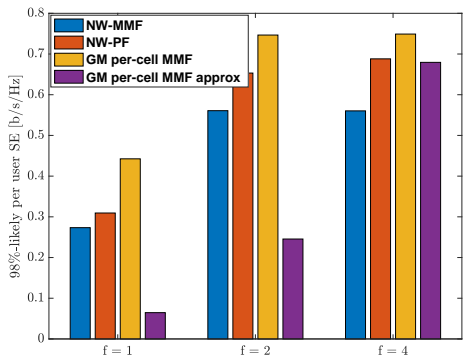}
	\caption{98\%-likely SE of CU $k$ for UL data transmission with correlated Rayleigh fading channel with different pilot reuse factors.}
	\label{fig:95likelyULbar}
\end{figure}

In Figs.\ref{fig:corr_dl_f1}-\ref{fig:corr_dl_f4}, we provide the CDF of the SE of all the users for DL data transmission with pilot reuse factors $1$, $2$ and $4$, respectively. As it can be seen from these figures, similar to the UL, the proposed GM of per-cell MMF SE provides higher SE for lower tail users and rest of the users are happier with NW-PF. 
\begin{figure*}[t]
	\centering
	\begin{subfigure}[b]{0.3\linewidth}
	\includegraphics[width=\linewidth]{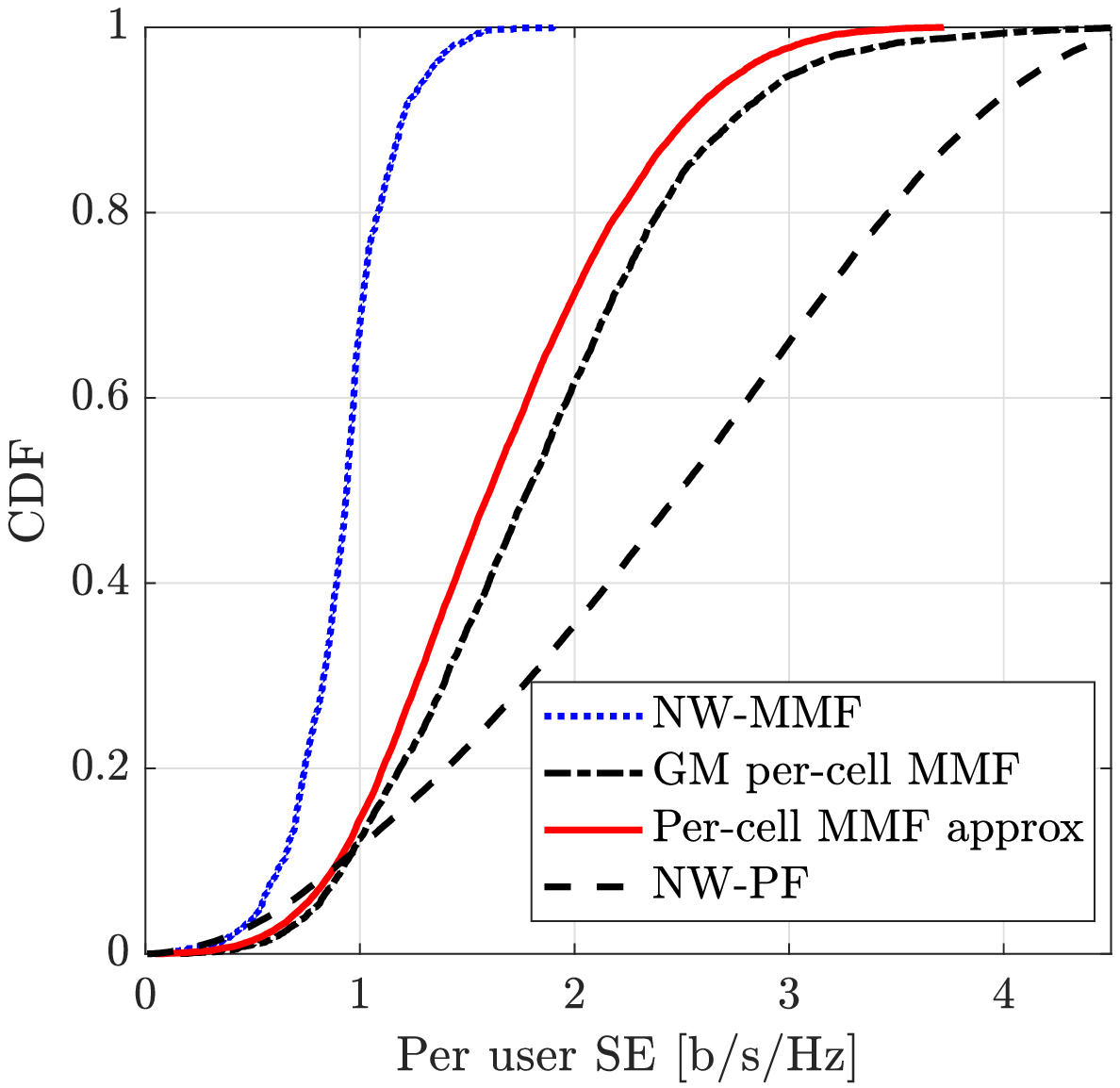}
	\caption{$f = 1$}
	\label{fig:corr_dl_f1}
     \end{subfigure}
 \begin{subfigure}[b]{0.3\linewidth}
	\includegraphics[width=\linewidth]{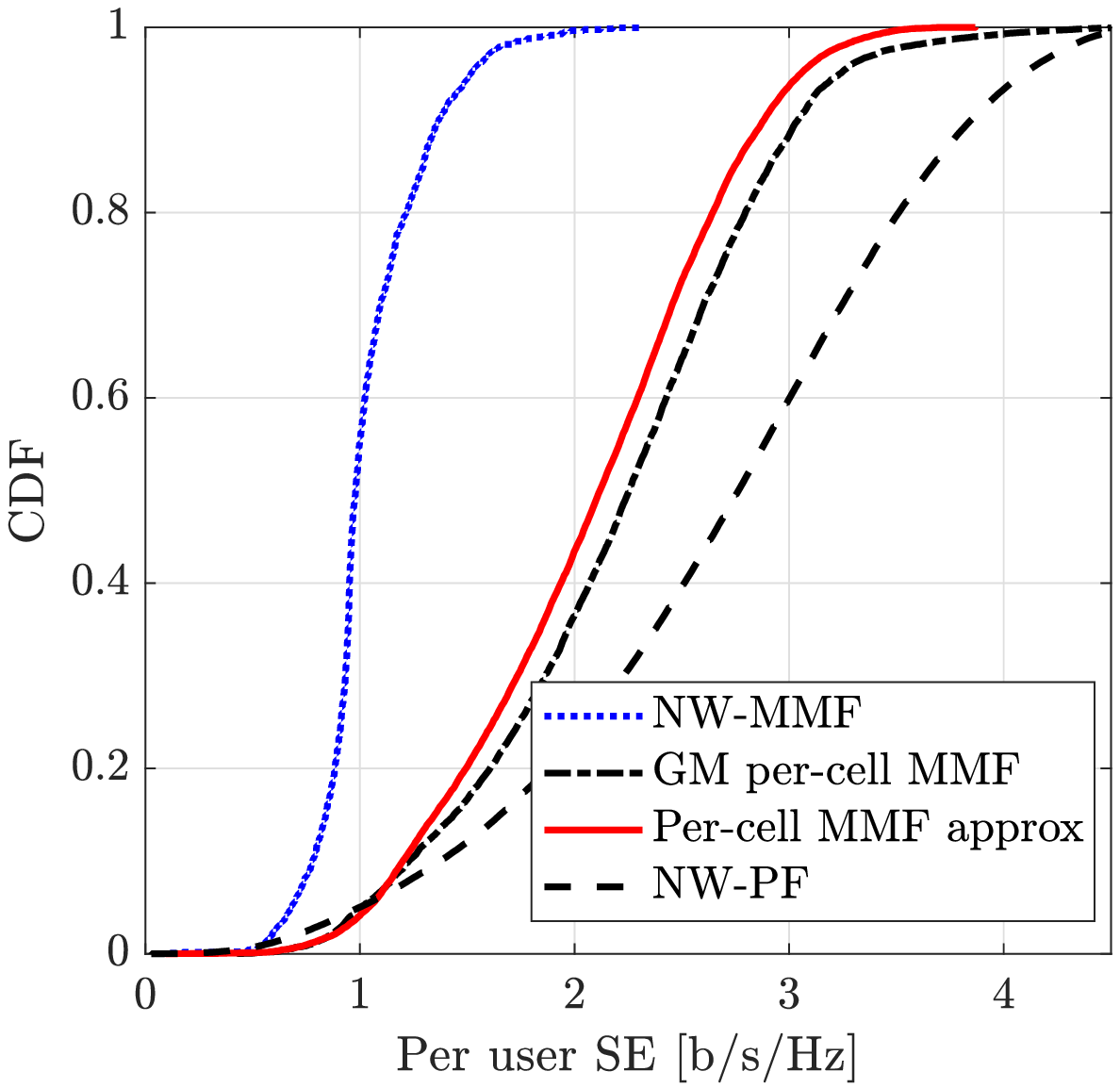}
	\caption{$f = 2$.}
	\label{fig:corr_dl_f2}
\end{subfigure}
\begin{subfigure}[b]{0.3\linewidth}
	\includegraphics[width=\linewidth]{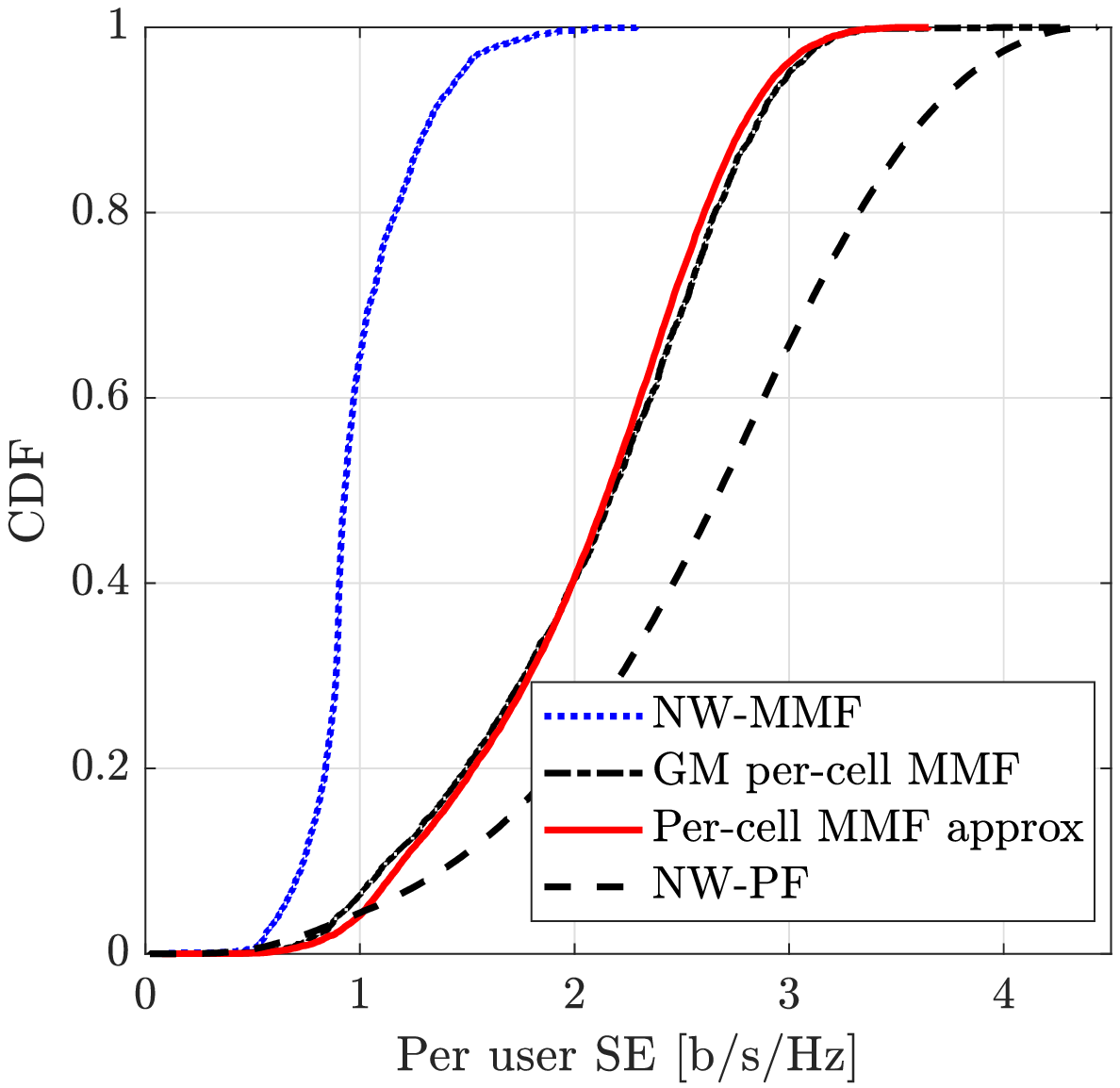}
	\caption{$f = 4$.}
	\label{fig:corr_dl_f4}
\end{subfigure}
\caption{SE of CU $k$ for DL data transmission with correlated Rayleigh fading channel with different pilot reuse factors.}
\label{fig:corr_dl}
\end{figure*}

In the case of reuse factor $1$, the $10\%$ weakest users get higher SE when using the proposed GM per-cell MMF SE. For the reuse factor $2$ and $4$, the $5\%$ and $2\%$ of weakest users get higher SE when using the proposed GM per-cell MMF SE, respectively. Fig. \ref{fig:98likelyDLbar} shows a bar diagram of the 98\%-likely SE of user $k$ for different reuse factors. It verifies the fairness performance of our proposed approach.
In addition, it can be seen from the figures that per-cell MMF approximate results are similar to our proposed GM per-cell MMF algorithm but our algorithm generally provides higher SE.
\begin{figure}[htb!]
	\centering
	\includegraphics[width=.8\columnwidth]{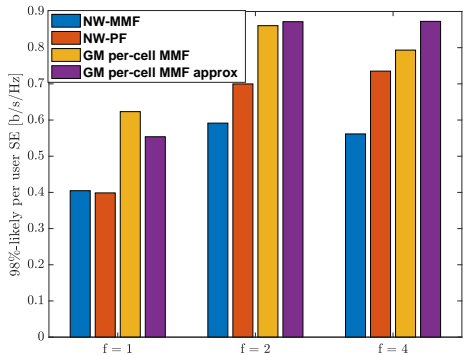}
	\caption{98\%-likely SE of CU $k$ for DL data transmission with correlated Rayleigh fading channel with different pilot reuse factors.}
	\label{fig:98likelyDLbar}
\end{figure}
{Table}~\ref{table:2}, provides the 95\%-likely sum SE for UL data transmission with correlated fading channel for pilot reuse factors $1$, $2$ and $4$. It can be seen from the results that for all cases, NW-PF scheme performs the best in terms of sum SE as it can be seen as an approximation to the sum SE maximization in the high SINR regime. In addition, per-cell MMF approximate solution provides similar results as our proposed method. In addition, we can see that the NW-MMF scheme has the lowest sum SE among all the schemes as expected.
\vspace*{0.1cm}
\begin{table}[htb!]
	\caption{
		The 95\%-likely sum SE for UL data transmission with correlated fading channel.
	}
	\centerline{
		\begin{tabular}{|c||c|c||c|}
			\hline
			\cline{2-4}
			&$f=1$ & $f=2$ & $f=4$ \\
			\hline\hline
			NW-MMF & $35.1$ & $53.9$ & $55.7$ \\
			\hline
			NW-PF &$ 175.5$& $203.9$ & $ 207.8$ \\
			\hline
			Per-cell MMF approx &$ 102.3$& $146.4$ & $ 162.2$ \\
			\hline 
			GM per-cell MMF &$106.5$& $140.6$ & $144.1$ \\
			\hline
	\end{tabular}}
	\label{table:2}
\end{table}
\vspace*{0.1cm}
In addition, we can see the 95\%-likely sum SE for DL data transmission with correlated fading channel for pilot reuse factors $1$, $2$ and $4$ in {Table}~\ref{table:3}. Similar to the UL case, the NW-PF provides the highest sum SE. In addition, our proposed GM per-cell MMF has similar performance as the approximate solution in which both of them offer higher sum SE than NW-MMF.

We observe that the results of the proposed GM per-cell MMF and the approximate solutions for UL and DL data transmission are very similar. It indicates that the approximate results are good approximates for our proposed scheme. In addition, the approximate solutions can be achieved by orders of magnitude faster than the proposed GM per-cell MMF. Therefore, the aforementioned discussions, highlight the potential of using the approximate solution as the benchmark results for comparison purposes. However, more accurate results are offered via GM per-cell MMF power control scheme at the cost of higher computational complexity.

\vspace{0.1cm}
\begin{table}[htb!]
	
	\caption{
		The 95\%-likely sum SE for DL data transmission with correlated fading channel.
	}
	\centerline{
		\begin{tabular}{|c||c|c||c|}
			\hline
			\cline{2-4}
			&$f=1$ & $f=2$ & $f=4$ \\
			\hline\hline
			NW-MMF & $42.7$ & $52.5$ & $51.4$ \\
			\hline
			NW-PF &$178.6$& $199.1$ & $195.1$ \\
			\hline
			Per-cell MMF approx &$110.6$& $143.6$ & $146.7$ \\
			\hline
			GM per-cell MMF &$124$& $149.4$ & $146$ \\
			\hline
	\end{tabular}}
	\label{table:3}
\end{table}

If we do a one-to-one comparison for all the users and calculate the percentage of users that get better SE with NW-MMF than with NW-PF or the proposed scheme, we get the results for the reuse factor $1$ provided in {Table}~\ref{table:1}. These numbers show that roughly one tenth of the users get higher SE, but we also see from the CDF curves that their SE gains are tiny, thus percentage values like this need to be taken with a grain of salt.

\begin{table}[htb!]
	\centering
	\caption{Percentage of users getting better SE using NW-MMF.} 
	\begin{center}
		\begin{tabular}{| c| c| c|}
			\hline
			&NW-PF& GM per-cell MMF \\ 
			\hline
			Uplink & 10\% & 12\% \\
			\hline
			Downlink & 11\% & 8\% \\
			\hline
		\end{tabular}\label{table:1}
	
	\end{center} 
\end{table}

\section{Conclusion}
In this paper, we analyzed different power control schemes that target fairness in multi-cell massive MIMO systems. 
We proposed to maximize the geometric mean of the per-cell max-min SEs. This approach is not subject to the same scalability issues as the conventional NW-MMF approach, which has received much attention in the literature. We solved the new problem formulation to global optimality and achieved better or comparable performance as the previous heuristic scheme in \cite{redbook} that also targeted to resolve the scalability issue of NW-MMF. Furthermore, our proposed approach provides more fairness towards weak users in comparison with NW-PF. The proposed solutions can be applied to many different channel models.

\appendix
\subsection{Proof of Lemma $3$}\label{proof1}

The function $\log \left(\log_2 \left( 1_{\epsilon}+ e^{x}\right)\right)$ is a concave function if and only if $f\left(x\right) = \log \left(\log \left( 1_{\epsilon}+ e^{x}\right)\right)$ is a concave function, since all logarithms are equal up to a scaling factor.

	The first derivative of $f\left(x\right)$ is
	\begin{equation}
	\begin{aligned}
	f'(x) = \frac{1}{\log\left(1_{\epsilon}+e^x\right)}\frac{1}{1_{\epsilon}+e^x}e^x = \frac{e^x}{1_{\epsilon}+e^x}\dfrac{1}{\log\left(1_{\epsilon}+e^x\right)},
	\end{aligned}
	\end{equation}
	and the second derivative can be written as
	\begin{equation}
	\begin{aligned}
	f''(x) &= \left(\frac{e^x}{1_{\epsilon}+e^x}\right)' \frac{1}{\log\left(1_{\epsilon}+e^x\right)} \\
	&+ \left(\frac{1}{\log\left(1_{\epsilon}+e^x\right)}\right)' \frac{e^x}{1_{\epsilon}+e^x}\\
	&= \frac{\left(1_{\epsilon}+e^x\right)\left(e^x\right)' - e^x\left(1_{\epsilon}+e^x\right)'}{\left(1_{\epsilon}+e^x\right)^2}\frac{1}{\log\left(1_{\epsilon}+e^x\right)} \\
	&+ \frac{-\frac{e^x}{1_{\epsilon}+e^x}}{\left(\log\left(1_{\epsilon}+e^x\right)\right)^2}\frac{e^x}{1_{\epsilon}+e^x}\\
	&= \frac{e^x + \left(e^x\right)^2 - \left(e^x\right)^2}{\left(1_{\epsilon}+e^x\right)^2 \log\left(1_{\epsilon}+e^x\right)} - \frac{\left(e^x\right)^2}{\left(1_{\epsilon}+e^x\right)^2 \left(\log\left(1_{\epsilon}+e^x\right)\right)^2}\\
	&= \left(1-\frac{e^x}{\log\left(1_{\epsilon}+e^x\right)}\right)\frac{e^x}{\left(1_{\epsilon}+e^x\right)^2\left(\log\left(1_{\epsilon}+e^x\right)\right)}.
	\end{aligned}
	\end{equation}
	Now we define $g\left(x\right)= e^x -\log\left(1+e^x\right)$ and we have
	\begin{equation}
	\begin{aligned}
	g'\left(x\right) &= e^x - \frac{e^x}{1_{\epsilon}+e^x} = e^x \left(1-\frac{1}{1_{\epsilon}+e^x}\right) \geq 0.
	\end{aligned}
	\end{equation}
This means that $g\left(x\right)$ is monotonically increasing in $x$, we also have $g\left(-\infty\right) = 0$. Therefore we have shown that $g\left(x\right)\geq 0, \forall x$. 
This implies $f''\left(x\right)\leq 0, \forall x$. Therefore, we have proved that $f\left(x\right)=\log\left(\log\left(1_{\epsilon}+e^x\right)\right)$ is a concave function in $x$.

\subsection{Power Budget Effects}
\label{budgetVSrate}
Fig. \ref{fig:budgetUL} and Fig. \ref{fig:budgetDL} show bar diagrams of the 95\%-likely sum SE for different power budget for both UL and DL, respectively. The power budgets are defined as the maximum transmit power of users and the maximum transmit power of BSs for UL and DL, respectively. The results are provided for a setup consisting of 4 cells, one BS per cell, and there are $K= 2$ users per cell. All other parameters have the same values as in Section \ref{result}. The results are consistent with those in Section \ref{result}, namely that NW-PF has the highest sum SE performance for different power budgets for both UL and DL. Our proposed approach is having higher performance than NW-MMF which has the lowest sum SE. The results also show that the power budget (max power of users in the UL case) has a higher impact on the UL data transmission than on the DL data transmission. Note that these results do not elaborate on the fairness level of the power control approaches because the individual SE performance of users is hidden. Hence, one should not design the network utility only based on the sum SE results.

\vspace*{0.2cm}
\begin{figure}[htb!]
	\centering
	\includegraphics[width=.8\columnwidth]{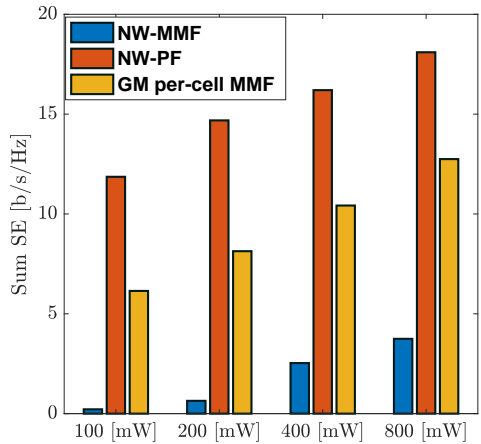}
	\caption{95\%-likely sum SE versus power budget of each user for UL data transmission with correlated Rayleigh fading channel and pilot reuse factors $f=1$.}
	\label{fig:budgetUL}
\end{figure}
\vspace*{0.2cm}
\begin{figure}[htb!]
	\centering
	\includegraphics[width=.8\columnwidth]{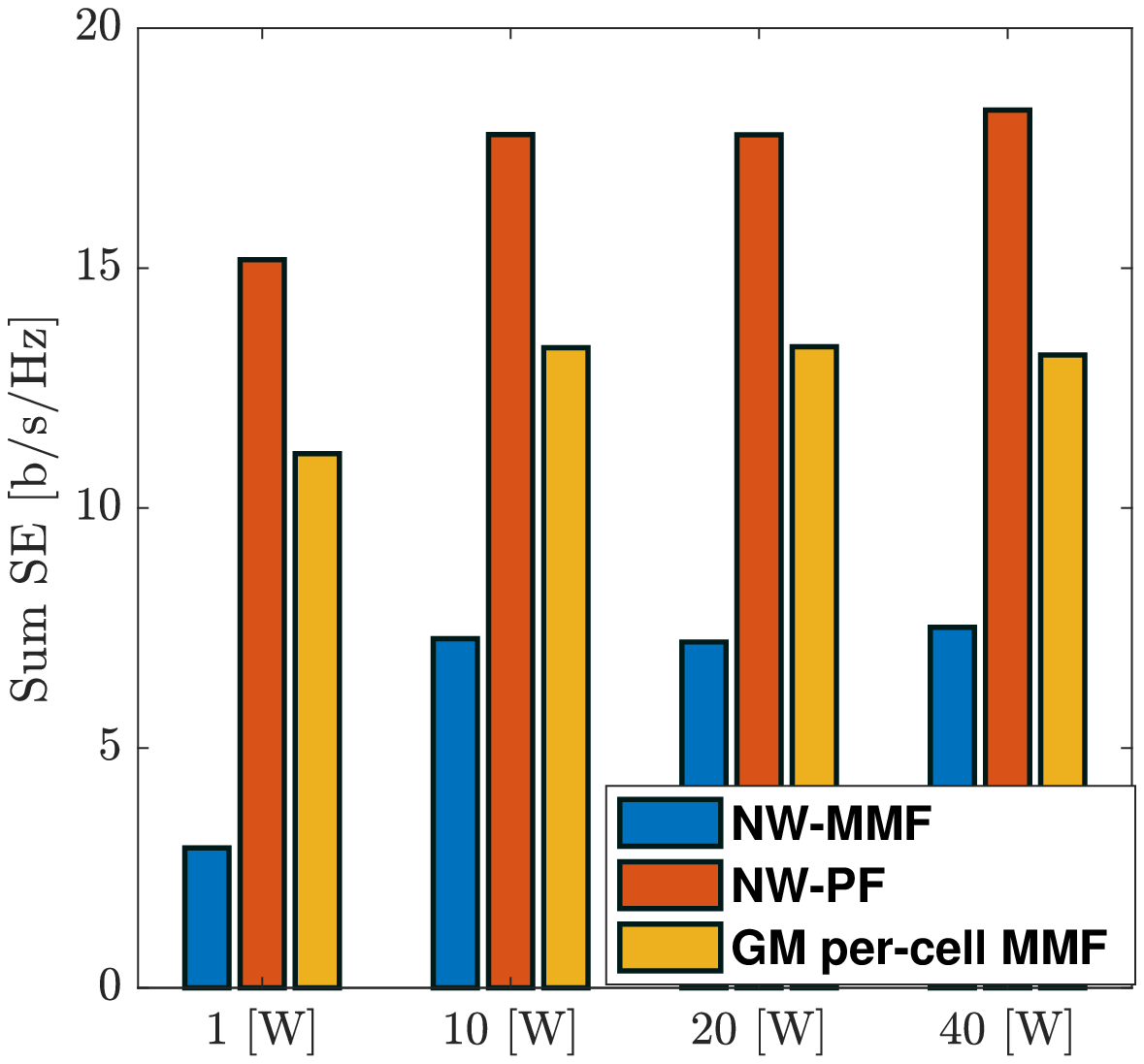}
	\caption{95\%-likely sum SE of versus power budget of each BS for DL data transmission with correlated Rayleigh fading channel and pilot reuse factors $f=1$.}
	\label{fig:budgetDL}
\end{figure}

\bibliographystyle{IEEEtran}
\bibliography{di}

\end{document}

%% file: commands.tex
\newtheorem{lemma}{Lemma}

\newcommand{\be}{\begin{equation}}
\newcommand{\ee}{\end{equation}}
\newcommand{\bal}{\begin{aligned}}
\newcommand{\eal}{\end{aligned}}	
\newcommand{\ba}{\begin{array}}
\newcommand{\ea}{\end{array}}
\newcommand{\bea}{\begin{eqnarray}}
\newcommand{\eea}{\end{eqnarray}}

\newcommand{\vbar}{\raisebox{.17ex}{\rule{.04em}{1.35ex}}}
\newcommand{\vbarind}{\raisebox{.01ex}{\rule{.04em}{1.1ex}}}
\newcommand{\R}{\ifmmode {\rm I}\hspace{-.2em}{\rm R} \else ${\rm I}\hspace{-.2em}{\rm R}$ \fi}
\newcommand{\T}{\ifmmode {\rm I}\hspace{-.2em}{\rm T} \else ${\rm I}\hspace{-.2em}{\rm T}$ \fi}
\newcommand{\N}{\ifmmode {\rm I}\hspace{-.2em}{\rm N} \else \mbox{${\rm I}\hspace{-.2em}{\rm N}$} \fi}
\newcommand{\B}{\ifmmode {\rm I}\hspace{-.2em}{\rm B} \else \mbox{${\rm I}\hspace{-.2em}{\rm B}$} \fi}
\newcommand{\Hil}{\ifmmode {\rm I}\hspace{-.2em}{\rm H} \else \mbox{${\rm I}\hspace{-.2em}{\rm H}$} \fi}
\newcommand{\C}{\ifmmode \hspace{.2em}\vbar\hspace{-.31em}{\rm C} \else \mbox{$\hspace{.2em}\vbar\hspace{-.31em}{\rm C}$} \fi}
\newcommand{\Cind}{\ifmmode \hspace{.2em}\vbarind\hspace{-.25em}{\rm C} \else \mbox{$\hspace{.2em}\vbarind\hspace{-.25em}{\rm C}$} \fi}
\newcommand{\Q}{\ifmmode \hspace{.2em}\vbar\hspace{-.31em}{\rm Q} \else \mbox{$\hspace{.2em}\vbar\hspace{-.31em}{\rm Q}$} \fi}
\newcommand{\Z}{\ifmmode {\rm Z}\hspace{-.28em}{\rm Z} \else ${\rm Z}\hspace{-.28em}{\rm Z}$ \fi}

\newcommand{\CN}{\mathcal{CN}}

\renewcommand{\theenumii}{.\arabic{enumii}}